\documentclass[11pt]{article}
\usepackage{amssymb}
\usepackage{amsfonts,amsmath, longtable}
        \def\theequation{\thesection.\arabic{equation}}

\newcommand{\tr}{{\rm tr}}

\newcommand{\mM}{{\mathcal M}}

\newcommand{\mR}{{\mathcal R}}
\newcommand{\mH}{{\mathcal H}}

\newcommand{\mO}{{\mathcal O}}

\newcommand{\mD}{{\mathcal D}}

\newcommand{\vf}{\varphi}
\newcommand{\al}{\alpha}

\newcommand{\vth}{\vartheta}

\newcommand{\Mat}{ {\rm Mat}(N,\mathbb C) }

\newcommand{\mC}{\mathbb C}
\newcommand{\mZ}{\mathbb Z}

\newcommand{\mS}{\mathcal S}

\newcommand{\z}{{\zeta}}

\newtheorem{lemma}{Lemma}%[section]

\newenvironment{proof}{\par\noindent{{\underline{\em Proof:}\ }}}{$\scriptstyle\blacksquare$}
\def\beq{\begin{equation}}
\def\eq{\end{equation}}
\def\p{\partial}

\newtheorem{theor}{Theorem}%[section]

\def\res{\mathop{\hbox{Res}}\limits}

\begin{document}

\setcounter{page}{1}

\begin{center}

\

%\vspace{-10mm}

{\Large{\bf  Non-ultralocal classical $r$-matrix structure }}

\vspace{3mm}

{\Large{\bf for 1+1 field analogue of elliptic Calogero-Moser model}}

% \vspace{3mm}

%{\Large{\bf and its multispin generalizations}}

 \vspace{12mm}

% {\Large {M. Matushko}}$\,^{\diamond}$
%\qquad\quad\quad
 {\Large {Andrei Zotov}}
% $\,^{\diamond\,\bullet}$

  \vspace{10mm}

%$\diamond$ --
 {\em Steklov Mathematical Institute of Russian
Academy of Sciences,\\ Gubkina str. 8, 119991, Moscow, Russia}

%%$*$ -- {\em Center for Advanced Studies, Skoltech, 143026, Moscow, Russia}

%$\bullet$ --
%{\em Institute for Theoretical and Mathematical Physics,\\ Lomonosov Moscow State University, Moscow, 119991, Russia}

%   \vspace{3mm}

 {\small\rm {e-mail: zotov@mi-ras.ru}}

\end{center}

\vspace{0mm}

\begin{abstract}
We consider 1+1 field generalization of the elliptic Calogero-Moser model. It is shown that the Lax connection
satisfies the classical non-ultralocal $r$-matrix structure of Maillet type.
Next, we consider 1+1 field analogue of the spin Calogero-Moser model and its multipole (or multispin) extension.
Finally, we discuss the field analogue of the classical IRF-Vertex correspondence, which relates utralocal and
non-ultralocal $r$-matrix structures.
\end{abstract}

\bigskip
%

%\newpage
\bigskip
%{\small{ \tableofcontents }}

%\bigskip

%\newpage

\section{Introduction}\label{sec1}
\setcounter{equation}{0}
We consider 1+1 field analogues of the classical many-body integrable systems of Calogero-Moser type.
At the level of classical mechanics the spinless model describes interaction of $N$ particles on a complex plane
generated by the following Hamiltonian \cite{Calogero2}:
  \beq\label{e01}
  \begin{array}{c}
  \displaystyle{
H^{\hbox{\tiny{CM}}}=\sum\limits_{i=1}^N\frac{p_i^2}{2}-c^2\sum\limits_{i>j}^N\wp(q_i-q_j)\,,
 }
 \end{array}
 \eq
where $c\in\mC$ is a coupling constant and $\wp$ is the Weierstrass $\wp$-function (see Appendix A).
The positions of particles
$q_i\in\mC$ and momenta $p_i\in\mC$ are canonically conjugated:
  \beq\label{e02}
  \begin{array}{c}
  \displaystyle{
 \{p_i,q_j\}=\delta_{ij}\,,\quad \{p_i,p_j\}=\{q_i,q_j\}=0\,.
 }
 \end{array}
 \eq
 The Hamiltonian provide equations of motion ${\dot f}=\{H,f\}$, which are represented in the
 Lax form
  \beq\label{e03}
  \begin{array}{c}
  \displaystyle{
 \dot{L}(z)+[L(z),M(z)]=0\,,\qquad L(z),M(z)\in\Mat
 }
 \end{array}
 \eq
with spectral parameter $z$ \cite{Kr}. Then  $\tr L^k(z)$ are generating functions (in $z$) of conservation laws, and they are in involution with respect to the Poisson brackets (\ref{e02}) $\{\tr L^k(z), \tr L^m(w)\}=0$ due to existence
of the classical
$r$-matrix \cite{Skl3,BradenSuz}\footnote{Let us mention that for skew-symmetric $r$-matrices
(\ref{e04}) takes the form $\{L_1(z),L_2(w)\}=[L_1(z)+L_2(w),r_{12}(z,w)]$.
The classical $r$-matrices under consideration, interrelations between them and the corresponding Yang-Baxter equations are reviewed in the Appendix B.}:
  \beq\label{e04}
  \begin{array}{c}
  \displaystyle{
\{L_1(z),L_2(w)\}=[L_1(z),r_{12}(z,w)]-[L_2(w),r_{21}(w,z)]\,.
 }
 \end{array}
 \eq
The system (\ref{e01}) is extended to the spin Calogero-Moser model \cite{GH} for which we use the classical description suggested by E. Billey, J. Avan and O. Babelon \cite{BAB}. In this case the phase space has a component related to the classical spin
variables $S_{ij}$, $i,j=1,...,N$ arranged into the matrix $S\in\Mat$. This part of the phase space is described as
a quotient space ${\mathcal O}//{\mathfrak H}$, where ${\mathcal O}$ is a coadjoint orbit of ${\rm GL}_N$ Lie group
acting on the Lie coalgebra ${\rm gl}_N^*$,
%with the Poisson-Lie brackets
 and ${\mathfrak H}\subset {\rm GL}_N$ is the Cartan subgroup of ${\rm GL}_N$. The spin model is easily defined
 on the unreduced phase space ${\mathcal O}$, where it is described by the Poisson-Lie brackets
  \beq\label{e05}
  \begin{array}{c}
  \displaystyle{
   \{S_{ij},S_{kl}\}=-S_{il}\delta_{kj}+S_{kj}\delta_{il}
 }
 \end{array}
 \eq
with the Casimir functions (eigenvalues of $S$) be fixed. Then the Lax equation (\ref{e03}) and
the classical $r$-matrix structure (\ref{e04}) acquire unwanted terms in their right-hand sides.
These unwanted terms vanish after performing the Hamiltonian reduction by the action of ${\mathfrak H}$.

Finally, the spin Calogero-Moser model can be generalized to the multispin case \cite{NN}. Then the
spin part of the phase space is given by $({\mathcal O}^1\times...\times {\mathcal O}^n)//{\mathfrak H}$.
The Lax equation and $r$-matrix structure is defined similarly to the spin case.
%This model was called in \cite{NN} as the elliptic Gaudin model since the corresponding Lax matrix
% (as a function of spectral parameter) has a set of simple poles on elliptic curve. Here we use different
% terminology. We call the Gaudin model those models which have ''pure spin'' phase space, and have no many-body %degrees of freedom.
 There are also multipole (Gaudin type) models which have ''pure spin'' phase space.
 In the elliptic case they were introduced in \cite{STS}. An important property of these
 models is that they are described by non-dynamical classical $r$-matrix, while $r$-matrices for
 the Calogero-Moser model and its generalizations are dynamical, that is they depend explicitly on dynamical variables.
A review on the discussed above
elliptic integrable systems can be found in \cite{ZS,TrZ1}.

An integrable field analogue of a finite-dimensional integrable system is some integrable infinite-dimensional generalization of the finite-dimensional model. One possibility is to consider a large $N$ limit. This yields collective field theories of hydrodynamical type. In this paper we use a different approach. A set of
coordinates $\varphi_1,...,\varphi_n$ on the phase space of the finite-dimensional model
is replaced with a set of field variables $\varphi_1(x),...,\varphi_n(x)$ depending on $x$ -- real valued variable considered as a coordinate on a line or a circle. For definiteness we consider $\varphi_1(x),...,\varphi_n(x)$
to be periodic functions on a circle: $\varphi_i(x)=\varphi_i(x+2\pi)$. Dynamics provides also dependence on the time variable $t$, that is
we deal with the fields depending on the time variable $t$ and the space variable $x$.
 Integrable field analogues of the finite-dimensional many-body systems are known beginning from 2d Toda field theory
\cite{Mikh}. For example, the 2-body closed Toda model is described by the ordinary differential equation
${\ddot q}+\frac12\sinh(q)=0$, $q=q(t)$. Its field analogue is the Sinh–Gordon (partial differential) equation
 $\p_t^2q-\p_x^2q+\frac12\sinh(q)=0$, $q=q(t,x)$. Another example is the Euler top in three-dimensional space:
 ${\dot {\vec S}}={\vec S}\times J({\vec S})$, where $\vec S=\vec S(t)=(S_1(t),S_2(t),S_3(t))$ is
 the angular momentum vector. Its field analogue is given by the Landau-Lifshitz equation:
  \beq\label{e055}
  \begin{array}{c}
  \displaystyle{
   \p_t{ {\vec S}}={\vec S}\times\p^2_x{\vec S} +{\vec S}\times J({\vec S})\,,
   \qquad \vec S=\vec S(t,x)\,,
 }
 \end{array}
 \eq
where $\vec S$ is now the magnetization vector in 1-dimensional model of (ferro)magnet.

Integrability in the 1+1 field case is a subtle question. In this paper by integrability of the field theory we mean
a generalization of the finite-dimensional construction based on the Lax equation (\ref{e03})
and the classical $r$-matrix structure (\ref{e04}). Namely, we assume that the
equations of motion are
represented in the form of the zero curvature equation known as the Zakharov-Shabat equation
\cite{ZaSh,FT}:
 \beq\label{e09}
 \begin{array}{c}
  \displaystyle{
 \partial_{t}{U}(z)-k\partial_{x}{V}(z)+[{U}(z), {V}(z)]=0\,,\qquad {U}(z), {V}(z)\in\Mat\,,
  }
 \end{array}
\eq
 In the limit to the finite-dimensional case (when all the fields become independent of $x$) the term $\partial_{x}{V}(z)$ vanishes,
and (\ref{e09}) becomes the Lax equation (\ref{e03}), where $L(z)$ is the limit of $U(z)$,
while $M(z)$ coincides with the limit of $V(z)$ up to some expression commuting with $L(z)$.
The conservation laws in the field case are generated by $\tr(T^k(z,2\pi))$, where $T(z,x)\in\Mat$
is the monodromy matrix
 \beq\label{e10}
 \begin{array}{c}
  \displaystyle{
 T(z,x)={\rm Pexp}\Big( \frac{1}{k}\int\limits_0^x dy\,  U(z,y) \Big)\,.
  }
 \end{array}
\eq
The sufficient condition for the
Poisson commutativity
%of Hamiltonians
$\{\tr(T^k(z,2\pi)),\tr(T^m(w,2\pi))\}=0$ was suggested by
J.-M. Maillet \cite{Maillet} and is know as the non-ultralocal Maillet bracket:
  \beq\label{e11}
  \begin{array}{c}
  \displaystyle{
\{U_1(z,x),U_2(w,y)\}=
  }
  \\ \ \\
  \displaystyle{
\Big(-k\p_x {\bf r}_{12}(z,w|x)
+ [U_1(z,x),{\bf r}_{12}(z,w|x)]-[U_2(w,y),{\bf r}_{21}(w,z|x)]\Big)\delta(x-y)-
  }
  \\ \ \\
  \displaystyle{
  -\Big({\bf r}_{12}(z,w|x)+{\bf r}_{21}(w,z|x)\Big)k\delta'(x-y)\,.
 }
 \end{array}
 \eq

In contrast to the Toda model the field generalization of the Calogero-Moser model is highly
non-trivial (see also \cite{ZZ} for the field analogue of the Ruijsenaars-Schneider model) even at the level
of the Hamiltonian. It is of the form
 \beq\label{e06}
 \begin{array}{c}
  \displaystyle{
  \mH^{\hbox{\tiny{2dCM}}}=\int\limits_{0}^{2\pi} d x H^{\hbox{\tiny{2dCM}}}(x)
  }
 \end{array}
\eq
with the density
 \beq\label{e07}
 \begin{array}{c}
  \displaystyle{
H^{\hbox{\tiny{2dCM}}}(x)=-\sum\limits_{j=1}^N p_j^2(kq_{j,x}+c)+\frac{1}{Nc}\Big(\sum\limits_{j=1}^N p_j(kq_{j,x}+c)\Big)^2
+\frac14\sum\limits_{j=1}^N\frac{k^4 q_{j,xx}^2}{kq_{j,x}+c}+
  }
  \\
    \displaystyle{
    +\frac12\sum\limits_{i\neq j}^N\Big( (kq_{i,x}+c)^2(kq_{j,x}+c)+(kq_{i,x}+c)(kq_{j,x}+c)^2
    -ck^2(q_{i,x}-q_{j,x})^2  \Big)\wp(q_i-q_j)+
      }
  \\
    \displaystyle{
    +\frac{k^3}{2}\sum\limits_{i\neq j}^N\Big( q_{i,x}q_{j,xx}-q_{i,xx}q_{j,x} \Big)\zeta(q_i-q_j)\,,
    }
 \end{array}
\eq
where $q_{i,x}=\p_xq_i(x)$ and $q_{i,xx}=\p^2_xq_i(x)$, and the set of fields\footnote{Similarly to the finite-dimensional case we consider all the fields to be $\mC$-valued.} is given by $q_1(t,x),...,q_N(t,x)$ and
$p_1(t,x),...,p_N(t,x)$ with the
canonical Poisson brackets
 \beq\label{e08}
 \begin{array}{c}
  \displaystyle{
\{p_i(x),q_j(y)\}=\delta_{ij}\delta(x-y)\,,\qquad \{p_i(x),p_j(y)\}=\{q_i(x),q_j(y)\}=0\,.
  }
 \end{array}
\eq
 The expression (\ref{e07}) for the field theory Hamiltonian  generalizes the finite-dimensional Hamiltonian (\ref{e01}). In (\ref{e07})
 $k\in\mR$ is a free parameter\footnote{One can put $k=1$ but we keep it as it is. Its degree shows the total
 order of derivative with respect to $x$ in the corresponding term.}.  When all the fields are independent of $x$
one can put $k=0$, and (\ref{e06}) turns into (\ref{e01}) with the factor $-2c$.
The field analogue of the Calogero-Moser model (\ref{e07}) was obtained for 2-body case in \cite{Krich2}
and \cite{LOZ} and then extended to $N$-body case in \cite{Krich22}.
The canonical Poisson brackets (\ref{e08})
 and the Hamiltonian (\ref{e07}) provide equations of motion, which are
represented in the form of the zero curvature equation known as the Zakharov-Shabat equation (\ref{e09}),
see details in \cite{Krich2,LOZ,Krich22}. At the same time the classical $r$-matrix structure for the model
(\ref{e07}) was unknown, and this is one of the purposes of the paper.

Integrable 1+1 field theories are actively studied at the classical and quantum levels due to their deep relation to gauge theories, strings and conformal field theories, existence of soliton solutions and many other remarkable properties
\cite{CWY,Vicedo,LOZ22,Lacroix1,Lacroix2,Bykov,CSV}.
We hope to include into consideration a large class of models obtained as 1+1 field generalizations of the widely known finite-dimensional integrable systems of Calogero-Ruijsenaars family.

%
%
%1+1 STS
%Feher-ZaTa
%

\paragraph{Purpose of the paper.}
%First, we prove (\ref{e11}) for the field Calogero-Moser model (\ref{e07}). Then we prove a similar result
%for its spin generalization. Finally, we show how the Maillet bracket (\ref{e11}) can be obtained
%starting from the ultralocal bracket related to pure spin models by using a continuous
%version of the classical IRF-Vertex relation.
%
The first result of the paper is that the field analogue of the
Calogero-Moser model (\ref{e07}) is described by the Maillet $r$-matrix structure
(\ref{e11}) with the classical $r$-matrix as in the finite-dimensional case \cite{Skl3,BradenSuz} but the positions
of particles $q_i$ become the fields $q_i(x)$.

Next, we discuss the 1+1 field generalizations of the spin Calogero-Moser model.
 There are at least three different ways to obtain 1+1 field analogues of the finite-dimensional integrable systems:
 from elliptic families of solutions to integrable hierarchies of KP or 2d Toda types,
 by taking a continuous limit from the classical integrable chain \cite{FT} and by the group-theoretical approach in the infinite-dimensional case. The spinless model (\ref{e07}) can be obtained using all three ways. The first
 method was used in \cite{Krich2,Krich22,ZZ}, the second one in \cite{ZZ} and the third one in \cite{LOZ,LOZ22}.
 In the last one approach based on the affine Higgs bundles
 the Lax matrix (the $U$-matrix) for the 1+1 spin (and multispin) Calogero-Moser
 model was derived. We use it to show that the spin model is also described by the $r$-matrix
 structure (\ref{e11}). Similarly to the finite-dimensional case its description includes
 an intermediate step. For the unreduced model the non-ultralocal term is absent but the right-hand side
 of (\ref{e11}) contains some unwanted term. The additional reduction removes the unwanted term but
 possibly adds the non-ultralocal term.

At the level of finite-dimensional models a wide class of models (including certain integrable tops, Gaudin models
and spin chains) are governed by
non-dynamical skew-symmetric $r$-matrices such as the Belavin-Drinfeld elliptic $r$-matrix \cite{BD}
and its trigonometric and rational degenerations. These are the ''pure spin'' models (without many-body degrees of freedom). Their 1+1 field generalizations
are the models of the Landau-Lifshitz-Heisenberg type \cite{Skl}, its higher rank versions \cite{GS,AtZ2}, and multi pole
generalizations known as the 1+1 (affine) Gaudin models \cite{Z11,Vicedo,LOZ22,Lacroix1}. The classical $r$-matrix structure
(\ref{e11}) is simplified since the derivative term ($\p_x {\bf r}_{12}$) and the non-ultralocal term
(the one with $\delta'$) are absent in these cases. The non-ultralocal models
are more complicated \cite{Bordemann,Bazhanov}. In the general case $r$-matrices are not skew-symmetric and depend on dynamical variables. In quantum statistical models the relation between certain dynamical and non-dynamical $R$-matrices was formulated by R.J. Baxter \cite{Baxter2} and is know as the IRF-Vertex correspondence.
Similar relation exists for the finite-dimensional integrable systems \cite{LOZ,ZS}, and it can be extended
to 1+1 field theories.  At the level of Lax pairs the IRF-Vertex transformation acts by some special gauge transformation\footnote{Similar relations for other type models are known from \cite{ZaTa,Feher}.}
 \beq\label{e111}
 \begin{array}{c}
  \displaystyle{
 U(z,x)\rightarrow G(z,x)U(z,x)G^{-1}(z,x)+k\p_xG(z,x)G^{-1}(z,x)\,,
  }
 \end{array}
\eq
which was treated in \cite{LOZ} as a modification of the affine Higgs  bundle.
This establishes gauge equivalence between the 1+1 Calogero-Moser model (\ref{e07}) and some
special Landau-Lifshitz type model. Examples can be found in \cite{AtZ1,AtZ2,AtZ3}.
The Landau-Lifshitz type model is described by ultralocal and non-dynamical $r$-matrix
of the Belavin-Drinfeld type. We will argue that the gauge transformation (\ref{e111})
given by the intertwiner matrix of the IRF-Vertex correspondence
provides the same relation between dynamical and non-dynamical $r$-matrices as
in the finite-dimensional case. In this way we show that the IRF-Vertex transformation
maps the 1+1 field Calogero-Moser model into the form with ultralocal $r$-matrix structure (related to Landau-Lifshitz model).

%we discuss gln case, and spin models char class, it is interesting to spinless related to other root systems

The paper is organized as follows. In the next Section we recall construction of the classical $r$-matrix
structures for the Calogero-Moser model and its spin (and multispin) generalization. In Section
\ref{sec3} we present field analogues of the finite-dimensional constructions from Section \ref{sec2}.
In Section \ref{sec4} the field version of the classical IRF-Vertex relation is discussed.
Elliptic functions definitions and properties are given in the Appendix A. A brief review of
$R$-matrices under consideration and interrelations between them including the IRF-Vertex relations
can be found in the Appendix B.

%%%%%%%%%%%%%%%%%%%%%%%%%%%%%%%%%%%%%%%%%%%%%%%%%%%%%%%%%%%%%%%%%%%%%%%%%%%%%%%%%%%%%%%%%%%
%%%%%%%%%%%%%%%%%%%%%%%%%%%%%%%%%%%%%%%%%%%%%%%%%%%%%%%%%%%%%%%%%%%%%%%%%%%%%%%%%%%%%%%%%%%
\section{Classical mechanics}\label{sec2}
\setcounter{equation}{0}

Here we review the classical $r$-matrix structures for the spinless and spin
Calogero-Moser models in the finite-dimensional case.

\subsection{Spinless Calogero-Moser model}
The Lax representation with spectral parameter (\ref{e03}) for the spinless elliptic Calogero-Moser model
is given by the pair of $N\times N$ matrices \cite{Kr}:
  \beq\label{e21}
  \begin{array}{c}
  \displaystyle{
L^{\hbox{\tiny{CM}}}_{ij}(z)=\delta_{ij}(p_i+{c} E_1(z))+{c}(1-\delta_{ij})\phi(z,q_{ij})\,,\quad q_{ij}=q_i-q_j\,,
 }
 \end{array}
 \eq
  \beq\label{e22}
  \begin{array}{c}
  \displaystyle{
M^{\hbox{\tiny{CM}}}_{ij}(z)=-{c}\delta_{ij} d_i
-{c}(1-\delta_{ij})f(z,q_{ij})\,,\quad d_i=\sum\limits_{k:k\neq i}^N
\wp(q_{ik})\,,
 }
 \end{array}
 \eq
where $\phi(z,q)$ is the elliptic Kronecker function (\ref{a01}), and the functions $f(z,q)$, $E_1(z)$
are defined in (\ref{a04})-(\ref{a05}). Equations of motion
  \beq\label{e23}
  \begin{array}{c}
  \displaystyle{
 {\dot q}_i=p_i\,,\qquad  {\dot p}_i={c}^2\sum\limits_{k:k\neq
 i}^N\wp'(q_{ik})\,,\quad i=1,...,N
 }
 \end{array}
 \eq
coming from (\ref{e01})-(\ref{e02}) are reproduced\footnote{Equations (\ref{e23}) are reproduced from
from the Lax equations (\ref{e03}) up to a transformation $p_i\rightarrow p_i+{\rm const}$ since it corresponds to $L^{\hbox{\tiny{CM}}}(z)\rightarrow L^{\hbox{\tiny{CM}}}(z)+{\rm const}1_N$ ($1_N$ -- the identity $N\times N$ matrix), and ${\rm const}1_N$ is cancelled out from both sides of
(\ref{e03}).} from the Lax equations (\ref{e03}). For the proof one should use the identity (\ref{a11}) for the diagonal part of (\ref{e03}), and the identity (\ref{a09}) for its off-diagonal part.
The Hamiltonian (\ref{e01}) comes from the computation of $\tr((L^{\hbox{\tiny{CM}}}(z))^2)/2$ through
the identity (\ref{a10}).

The Poisson commutativity of higher Calogero-Moser
 Hamiltonians follows from the existence of the classical $r$-matrix satisfying
(\ref{e04}) for the Lax matrix (\ref{e21}).
In (\ref{e04}) the standard tensor notations are used: $L_1=L\otimes 1_N$, $L_2=1_N\otimes L$, where $1_N$ -- the identity
 $N\times N$ matrix. The permutation of $r$-matrix indices $12$ or $21$ means the permutation of the tensor components.
  \beq\label{e24}
  \begin{array}{c}
  \displaystyle{
 r_{12}(z,w)=\sum\limits_{i,j,k,l=1}^N E_{ij}\otimes E_{kl}\, r_{ij,kl}(z,w)\,,\qquad
 r_{21}(z,w)=\sum\limits_{i,j,k,l=1}^N E_{kl}\otimes E_{ij}\, r_{ij,kl}(z,w)\,,
 }
 \end{array}
 \eq
where $E_{ij}$ is the basis matrix units in $\Mat$, that is $(E_{ij})_{ab}=\delta_{ia}\delta_{jb}$.
The left-hand side of (\ref{e04}), by definition, is
  \beq\label{r05}
  \begin{array}{c}
  \displaystyle{
\{L_1(z),L_2(w)\}=\sum\limits_{i,j,k,l=1}^N \{L_{ij}(z),L_{kl}(w)\}\, E_{ij}\otimes
E_{kl}\,.
 }
 \end{array}
 \eq
Explicit expression for the classical $r$-matrix of the spinless Calogero-Moser model was suggested in
\cite{Skl3,BradenSuz}:
\beq\label{e26}
  \begin{array}{c}
  \displaystyle{
 r^{\hbox{\tiny{CM}}}_{12}(z,w)=
  }
  \\ \ \\
   \displaystyle{
 =
 (E_1(z-w)+E_1(w))\,\sum\limits_{i=1}^NE_{ii}\otimes E_{ii}+
 \sum\limits^N_{i\neq j}\phi(z-w,q_{ij})\,E_{ij}\otimes E_{ji}
 -\sum\limits^N_{i\neq j}\phi(-w,q_{ij})\,E_{ii}\otimes E_{ji}\,.
 }
 \end{array}
 \eq
Verification of this statement is by direct calculation. For example, the
left-hand side of (\ref{e04}) is
 \beq\label{e27}
  \begin{array}{c}
  \displaystyle{
\{L_{ij}(z),L_{kl}(w)\}=\delta_{ij}(1-\delta_{kl})(\delta_{ik}-\delta_{il})f(w,q_{kl})-
\delta_{kl}(1-\delta_{ij})(\delta_{ki}-\delta_{kj})f(z,q_{ij})\,,
 }
 \end{array}
 \eq
 or
  \beq\label{e28}
  \begin{array}{c}
  \displaystyle{
\{L_1(z),L_2(w)\}=\sum\limits^N_{i\neq j} E_{ii}\otimes E_{ij}\, f(w,q_{ij})-\sum\limits^N_{i\neq j}
E_{jj}\otimes E_{ij}\, f(w,q_{ij})-
 }
\\
  \displaystyle{
-\sum\limits^N_{i\neq j} E_{ij}\otimes E_{ii}\, f(z,q_{ij})+\sum\limits^N_{i\neq j} E_{ij}\otimes E_{jj}\,
f(z,q_{ij})\,.
 }
 \end{array}
 \eq
 The right-hand side of (\ref{e04}) is simplified through (\ref{a07}) and (\ref{a08}).

\subsection{Spin and multispin Calogero-Moser model}

\paragraph{Classical spin Calogero-Moser model: Lax equations.} The spin Calogero-Moser model was first introduced
at quantum level \cite{GH}. Its classical analogue was described in \cite{BAB} and we follow this description below.

Besides the
many-body degrees of freedom (\ref{e02}) its phase space has a ''classical spin'' component. The dynamical
spin variables $S_{ij}$, $i,j=1,...,N$ (for ${\rm gl}_N$ case) are naturally arranged into a matrix $S\in\Mat$.
The pair of $N\times N$ matrices
  \beq\label{e29}
  \begin{array}{c}
  \displaystyle{
L^{\hbox{\tiny{spin}}}_{ij}(z)=\delta_{ij}(p_i+S_{ii}E_1(z))+(1-\delta_{ij})S_{ij}\phi(z,q_{ij})\,,
 }
 \end{array}
 \eq
  \beq\label{e30}
  \begin{array}{c}
  \displaystyle{
  M^{\hbox{\tiny{spin}}}_{ij}(z)=-(1-\delta_{ij})S_{ij}f(z,q_i-q_j)
 }
 \end{array}
 \eq
does not satisfy the Lax equation (\ref{e03}) but satisfies the Lax equation with ''unwanted term'' in the right
hand side:
  \beq\label{e31}
  \begin{array}{c}
  \displaystyle{
 \dot{L}^{\hbox{\tiny{spin}}}(z)+[L^{\hbox{\tiny{spin}}}(z),M^{\hbox{\tiny{spin}}}(z)]=-\sum\limits_{i,j=1}^N E_{ij}\,(S_{ii}-S_{jj})S_{ij}E_1(z)f(z,q_{ij})\,.
 }
 \end{array}
 \eq
It provides equations of motion
  \beq\label{e32}
  \begin{array}{c}
  \displaystyle{
{\dot q}_i=p_i\,,\quad {\dot p}_i=\sum\limits_{j:j\neq
i}^NS_{ij}S_{ji}\wp'(q_i-q_j)
 }
 \end{array}
 \eq
and
  \beq\label{e33}
  \begin{array}{c}
  \displaystyle{
{\dot S}_{ii}=0\,,\quad {\dot S}_{ij}=\sum\limits_{k\neq i,j}^N
S_{ik}S_{kj}(\wp(q_i-q_k)-\wp(q_j-q_k))\,,\ i\neq j\,.
 }
 \end{array}
 \eq
These equations are Hamiltonian with the Hamiltonian function
  \beq\label{e34}
  \begin{array}{c}
  \displaystyle{
H^{\hbox{\tiny{spin}}}=\sum\limits_{i=1}^N\frac{p_i^2}{2}-\sum\limits_{i>j}^N S_{ij}S_{ji}\wp(q_i-q_j)\,.
 }
 \end{array}
 \eq
 The Poisson brackets for the spin variables are given by the Poisson-Lie structure (\ref{e05})
 on the Lie coalgebra ${\rm gl}_N^*$. Eigenvalues of matrix $S$ are the Casimir functions of
 (\ref{e05}), so that we may restrict ourselves on a coadjoint orbit $\mO$, and the level
 of the fixed eigenvalues defines its dimension. For example, in the generic case (when all eigenvalues are
 pairwise distinct) the dimension is maximal ${\rm dim}\mO^{\hbox{\tiny{max}}}=N(N-1)$, and
the orbit of minimal dimension
  \beq\label{e35}
  \begin{array}{c}
  \displaystyle{
{\rm dim}\mO^{\hbox{\tiny{min}}}=2(N-1)
 }
 \end{array}
 \eq
 corresponds
to $N-1$ coincident eigenvalues of $S$.

The unwanted term in (\ref{e31}) vanishes on the constraints
  \beq\label{e36}
  \begin{array}{c}
  \displaystyle{
 S_{ii}={c}={\rm const}\,,\quad i=1,...,N\,.
 }
 \end{array}
 \eq
They are in the agreement with the equations ${\dot S}_{ii}=0$ from (\ref{e33}). These are the first class constrains
generated by the coadgoint action of the Cartan subgroup ${\mathfrak H}\subset{\rm GL}_N$
 and
they should be supplemented by another\footnote{The number of the first class constraints
(\ref{e36}) is $N-1$ since $\tr S$ is the Casimir function.} $N-1$ gauge fixing conditions $\varsigma_k$, $k=1,...,N-1$. Finally, we have $2N-2$ second class constraints $\chi_a$ including (\ref{e36}) and some $\varsigma_k$.
Then one can use the Dirac formula\footnote{In (\ref{e37}) $\chi$ is a $2N-2$ dimensional row of constraints,
$\chi^T$ -- the corresponding column,
$||\{\chi_a,\chi_b\} ||^{-1}$ is the inverse of matrix between constraints and $f_{1,2}$ is a pair of functions
on $\mM^{\hbox{\tiny{spin}}}$.}
  \beq\label{e37}
  \begin{array}{c}
  \displaystyle{
  \{f_1,f_2\}\Big|_{\rm red}=\Big( \{f_1,f_2\}-\{f_1,\chi\}||\{\chi_a,\chi_b\} ||^{-1}\{\chi^T,f_2\} \Big)\Big|_{\rm on\ shell}
 }
 \end{array}
 \eq
to calculate the Poisson brackets on the reduced phase space
$\mM^{\hbox{\tiny{spin}}}=\mO//{\mathfrak H}$ -- the spin component of the phase space, which dimension is equal to
  \beq\label{e38}
  \begin{array}{c}
  \displaystyle{
{\rm dim}\mM^{\hbox{\tiny{spin}}}={\rm dim}\mO-2(N-1)\,.
 }
 \end{array}
 \eq
The resultant brackets depend on the choice of the gauge fixation $\varsigma_k$, and this effects
the equations of motion since the Hamiltonian equations ${\dot f}=\{H,f\}$ are now defined by the Poisson
brackets  (\ref{e37}) instead of (\ref{e05}) for the unreduced model.
At the same time,
regardless of the choice of gauge fixing conditions the equation (\ref{e31}) turns under reduction into the
 Lax equation (\ref{e03}).
 %The reduced Lax pair is obtained from (\ref{e29})-(\ref{e30}) by the gauge
 %transformation
%   with $g\in{\mathfrak H}$
% corresponding to the gauge fixation (in this case $g$ is independent of $z$), and restricting the result on-shell %the constraints.

\paragraph{Spin Calogero-Moser model: classical $r$-matrix.}
Let us define the classical $r$-matrix on the unreduced phase
space.  For the unreduced Lax matrix (\ref{e29})
and the set of Poisson brackets (\ref{e02}), (\ref{e05}) we have
the classical exchange relations with unwanted term
  \beq\label{e39}
  \begin{array}{c}
  \displaystyle{
\{L_1^{\hbox{\tiny{spin}}}(z),L_2^{\hbox{\tiny{spin}}}(w)\}=
[L^{\hbox{\tiny{spin}}}_1(z)+L^{\hbox{\tiny{spin}}}_2(w),r^{\hbox{\tiny{spin}}}_{12}(z-w)]-
  }
  \\ \ \\
  \displaystyle{
-\sum\limits_{k\neq l}^N E_{kl}\otimes E_{lk}\Big(S_{kk}-S_{ll}\Big)f(z-w,q_k-q_l)\,,
 }
 \end{array}
 \eq
 where
\beq\label{e40}
  \begin{array}{c}
  \displaystyle{
 r^{\hbox{\tiny{spin}}}_{12}(z,w)=r^{\hbox{\tiny{spin}}}_{12}(z-w)=
 E_1(z-w)\,\sum\limits_{i=1}^NE_{ii}\otimes E_{ii}+
 \sum\limits^N_{i\neq j}\phi(z-w,q_{ij})\,E_{ij}\otimes E_{ji}\,.
 }
 \end{array}
 \eq
 Due to (\ref{a01}) and (\ref{a06}) this classical $r$-matrix is skew-symmetric\footnote{Due to skew-symmetry the term $-[L_2(w),r_{21}(w,z)]$
 in (\ref{e04}) is written as $[L_2(w),r_{12}(z,w)]$ in (\ref{e39}).}
\beq\label{e41}
  \begin{array}{c}
  \displaystyle{
 r^{\hbox{\tiny{spin}}}_{12}(z,w)=-r^{\hbox{\tiny{spin}}}_{21}(w,z)\,.
 }
 \end{array}
 \eq
 The reduction with respect to the coadjoint action of ${\mathfrak H}$ kills the unwanted term.
  More precisely, we may restrict ourselves to the level of the first class constraints (\ref{e36}). Then
  the involution property $\{\tr((L^{\hbox{\tiny{spin}}}(z))^k),\tr((L^{\hbox{\tiny{spin}}}(w))^m)\}=0$
  holds true since these are the functions invariant with respect to the coadjoint action. This commutativity
  is independent of the gauge fixing conditions $\varsigma_k$. However, the Lax matrices themselves are
  not invariant functions. Thus, in order to have some explicit
  expression for the classical $r$-matrix on the reduced phase space one should fix $\varsigma_k$ and perform
  either Hamiltonian reduction or the Poisson reduction (\ref{e37}).
  It was argued in \cite{BDOZ}
 that the Poisson reduction via the Dirac bracket (\ref{e37}) saves the $r$-matrix form.

 %Alternatively, the reduction can be performed through the Dirac bracket (\ref{e37}).

\paragraph{Reduction to the spinless model.} Following \cite{BAB} let us show
how the described reduction provides the spinless Calogero-Moser model from the spin one.
Consider the unreduced spin model in the case of the minimal coadjoint orbit (\ref{e35}).
Then due to (\ref{e38}) ${\rm dim}\mM^{\hbox{\tiny{spin}}}=0$, i.e. the reduction
kills all spin degrees of freedom. Put it differently, the spinless Calogero-Moser model
corresponds to the spin model with the minimal orbit.

In the minimal orbit
case
$N-1$ eigenvalues of the matrix $S$ coincide. Without loss of generality we put them be equal to zero.
Therefore, $S$ is a rank one matrix
\beq\label{e44}
  \begin{array}{c}
  \displaystyle{
 S_{ij}=\xi_i\eta_j\,.
 }
 \end{array}
 \eq
The Poisson-Lie brackets (\ref{e05}) on ${\rm gl}_N^*$ are naturally parameterized by the canonical brackets
\footnote{The parameterization (\ref{e44})-(\ref{e45}) is a particular case of $M=1$ in the quiver like
parameterization $S_{ij}=\sum_{a=1}^M\xi^a_i\eta^a_j$ with $\{\xi^a_i,\eta^b_j\}=\delta_{ij}\delta^{ab}$
describing embedding of ${\rm gl}_N^*$ into $T^*\mC^{NM}$.}
\beq\label{e45}
  \begin{array}{c}
  \displaystyle{
 \{\xi_i,\eta_j\}=\delta_{ij}\,.
 }
 \end{array}
 \eq
Perform the gauge transformation
  \beq\label{e381}
  \begin{array}{c}
  \displaystyle{
L(z)\rightarrow g^{-1}(z)L(z)g(z)\,,\qquad
M(z)\rightarrow g^{-1}(z)M(z)g(z)-g^{-1}(z){\dot g}(z)
 }
 \end{array}
 \eq
 with
\beq\label{e46}
  \begin{array}{c}
  \displaystyle{
 g={\rm diag}(\xi_1,...,\xi_N)\,.
 }
 \end{array}
 \eq
 Then the matrix (\ref{e44}) transforms as $S_{ij}\rightarrow \xi_j\eta_j=S_{jj}$, and after restriction to
 the constraints (\ref{e36}) we get the Lax matrix of the spinless Calogero-Moser model (\ref{e21}).

 At the level of $r$-matrix a gauge transformation acts as
\beq\label{e42}
  \begin{array}{c}
  \displaystyle{
 r_{12}(z,w)\longrightarrow
 }
 \end{array}
 \eq
 $$
 g_1^{-1}(z)g_2^{-1}(w)\Big(
 r_{12}(z,w)+\{g_1(z),L_2(w)\}g_1^{-1}-\frac12\,\Big[\{g_1(z),g_2(w)\}g_1^{-1}(z)g_2^{-1}(w),L_2(w)\Big]
 \Big)g_1(z)g_2(w)\,.
 $$
 In our case $g$ is independent of spectral parameter. Also,
 % $g\in{\mathfrak H}$, that is
  $\{g_1,g_2\}=0$.
 Therefore, the gauge transformation maps the unreduced $r$-matrix (\ref{e40}) to
\beq\label{e43}
  \begin{array}{c}
  \displaystyle{
 r^{\hbox{\tiny{spin}}}_{12}(z-w)\longrightarrow
 %\stackrel{\rm reduction}{\longrightarrow}
 %\Big(
 g_1^{-1}g_2^{-1}r^{\hbox{\tiny{spin}}}_{12}(z-w)g_1g_2+g_1^{-1}g_2^{-1}\{g_1,L_2(w)\}g_2\,.
 %\Big)\Big|_{\rm on\ shell}\,.
 }
 \end{array}
 \eq
 It is easy to see from the upper expression that while the unreduced $r$-matrix is skew-symmetric and depends on the
 difference of spectral parameters, the reduced $r$-matrix looses these properties.

The result of the gauge transformation (\ref{e46}) is easy to calculate.
 Since
\beq\label{e47}
  \begin{array}{c}
  \displaystyle{
 g_1^{-1}g_2^{-1}r^{\hbox{\tiny{spin}}}_{12}(z-w)g_1g_2=r^{\hbox{\tiny{spin}}}_{12}(z-w)
 }
 \end{array}
 \eq
and
\beq\label{e48}
  \begin{array}{c}
  \displaystyle{
 g_1^{-1}g_2^{-1}\{g_1,L_2(w)\}g_2=
   }
  \\
  \displaystyle{
 =\sum\limits_{a,b=1}^N E_{aa}\otimes E_{bb}\,
 E_1(w)\{\xi_a,\xi_b\eta_b\}\xi_b^{-1}+
 \sum\limits_{a;\, b\neq c}^N E_{aa}\otimes E_{bc}\,\phi(w,q_{bc})\{\xi_a,\xi_b\eta_c\}\xi_a^{-1}\xi_b^{-1}\xi_c^{-1}
 }
  \\
  \displaystyle{
 =E_1(w)\sum\limits_{a=1}^N E_{aa}\otimes E_{aa}\,
 +
 \sum\limits_{a,b=1}^N E_{aa}\otimes E_{ba}\,\phi(w,q_{ba})=
 }
  \\
  \displaystyle{
 =
 E_1(w)\sum\limits_{a=1}^N E_{aa}\otimes E_{aa}\,
 -
 \sum\limits_{a,b=1}^N E_{aa}\otimes E_{ba}\,\phi(-w,q_{ab})\,,
 }
 \end{array}
 \eq
 one gets the $r$-matrix (\ref{e26}).

\paragraph{Nekrasov's multispin Calogero-Moser model.}
The multinspin extension of the spin Calogero-Moser model was introduced in \cite{NN}.
At the level of the Lax matrix (\ref{e29}) for unreduced model this is a generalization to
multi-pole case:
  \beq\label{e49}
  \begin{array}{c}
  \displaystyle{
L^{\hbox{\tiny{spins}}}_{ij}(z)=\delta_{ij}\Big(p_i+\sum\limits_{a=1}^M S^a_{ii}E_1(z-z_a)\Big)+(1-\delta_{ij})\sum\limits_{a=1}^M S^a_{ij}\phi(z-z_a,q_{ij})\,,
 }
 \end{array}
 \eq
The spin part of the phase space is given by the set of variables $S^a_{ij}$, $i,j=1,...,N$, $a=1,...,M$,
and
  \beq\label{e50}
  \begin{array}{c}
  \displaystyle{
\res\limits_{z=z_a}L^{\hbox{\tiny{spins}}}(z)=S^a\in\Mat\,,\quad a=1,...,M
 }
 \end{array}
 \eq
This space is equipped with the Poisson structure being a direct sum of (\ref{e05}):
  \beq\label{e51}
  \begin{array}{c}
  \displaystyle{
\{S^a_{ij}, S^b_{kl}\} = \delta^{ab} ( - S^a_{il}\delta_{kj}+S^a_{kj}\delta_{il} ) \,.
 }
 \end{array}
 \eq
By fixing the levels of the Casimir functions (eigenvalues of all $S^a$) the phase space
takes the form of a direct product of coadjoint orbits $\mO^1\times...\times \mO^M$.
The classical $r$-matrix structure for this unreduced models is similar to (\ref{e39}):
  \beq\label{e52}
  \begin{array}{c}
  \displaystyle{
\{L_1^{\hbox{\tiny{spins}}}(z),L_2^{\hbox{\tiny{spins}}}(w)\}=
[L^{\hbox{\tiny{spins}}}_1(z)+L^{\hbox{\tiny{spins}}}_2(w),r^{\hbox{\tiny{spin}}}_{12}(z-w)]+
  }
  \\ \ \\
  \displaystyle{
-\sum\limits_{k\neq l}^N\sum\limits_{a=1}^M E_{kl}\otimes E_{lk}\Big(S^a_{kk}-S^a_{ll}\Big)f(z-w,q_k-q_l)
 }
 \end{array}
 \eq
 with the same $r$-matrix (\ref{e40}). The Cartan subgroup $\mathfrak H$ acts on all $S^a$ simultaneously,
 thus providing the set of (moment map) constraints:
  \beq\label{e53}
  \begin{array}{c}
  \displaystyle{
 \sum\limits_{a=1}^M S^a_{kk}={\rm const}\,,\quad k=1,...,N\,.
 }
 \end{array}
 \eq
Similarly to the spin Calogero-Moser model one should supply
(\ref{e53}) with some gauge fixation and perform the reduction to the space
$(\mO^1\times...\times \mO^M)//{\mathfrak H}$.

%%%%%%%%%%%%%%%%%%%%%%%%%%%%%%%%%%%%%%%%%%%%%%%%%%%%%%%%%%%%%%%%%%%%%%%%%%%%%%%%%%%%%%%%%%%
%%%%%%%%%%%%%%%%%%%%%%%%%%%%%%%%%%%%%%%%%%%%%%%%%%%%%%%%%%%%%%%%%%%%%%%%%%%%%%%%%%%%%%%%%%%
\section{1+1 field theories}\label{sec3}
\setcounter{equation}{0}

Let us describe the field analogues for the models from the previous Section.

\subsection{1+1 Calogero-Moser field theory}

\paragraph{Equations of motion.} Following \cite{Krich22} let us write down equations of motion.
We use slightly different notations, and some signs differ from those in \cite{Krich22}.
 The Hamiltonian (\ref{e06})-(\ref{e07})
and the Poisson brackets (\ref{e08}) provides the following equations:
  \beq\label{e54}
  \begin{array}{c}
  \displaystyle{
 {\dot q}_i\equiv\{H^{\hbox{\tiny{2dCM}}},q_i\}=
 -2(kq_{i,x}+c)\Big(p_i-\frac{1}{Nc}\sum\limits_{j=1}^Np_j(kq_{j,x}+c)\Big)\,,
 }
 \end{array}
 \eq
  \beq\label{e55}
  \begin{array}{c}
  \displaystyle{
 {\dot p}_i\equiv\{H^{\hbox{\tiny{2dCM}}},p_i\}=
 -k\p_x\bigg( p_i^2-2p_i\frac{1}{Nc}\sum\limits_{j=1}^Np_j(kq_{j,x}+c)+
 \frac{1}{2}\frac{k^3 q_{i,xxx}}{kq_{i,x}+c}-\frac{1}{4}\frac{k^4 q^2_{i,xx}}{(kq_{i,x}+c)^2} \bigg)-
  }
  \\
  \displaystyle{
-2\sum\limits_{j:j\neq i}^N \bigg(
 (kq_{j,x}+c)^3\wp'(q_{ij})-3k^2(kq_{j,x}+c)q_{j,xx}\wp(q_{ij})-k^3q_{j,xxx}\zeta(q_{ij})
 \bigg)\,.
 }
 \end{array}
 \eq
 Introduce notations
  \beq\label{e56}
  \begin{array}{c}
  \displaystyle{
 \al_i=(kq_{i,x}+c)^{1/2}\,, \quad i=1,...,N
 }
 \end{array}
 \eq
and
  \beq\label{e57}
  \begin{array}{c}
  \displaystyle{
  \kappa=-\frac{1}{Nc}\sum\limits_{j=1}^Np_j(kq_{j,x}+c)=-\frac{1}{Nc}\sum\limits_{j=1}^Np_j\al_j^2\,.
 }
 \end{array}
 \eq
Then the Hamiltonian (\ref{e07}) and the equations of motion (\ref{e54})-(\ref{e55}) are written in a more
compact form:
 \beq\label{e58}
 \begin{array}{c}
  \displaystyle{
H^{\hbox{\tiny{2dCM}}}(x)=-\sum\limits_{i=1}^N p_i^2\al_i^2+Nc\kappa^2
+\sum\limits_{i=1}^N k^2\al_{i,x}^2+
  }
  \\
  \displaystyle{
+\frac{k}{2}\sum\limits_{i\neq j}^N\Big( \al_i\al_{j,x}-\al_j\al_{i,x}+c(\al_{i,x}-\al_{j,x}) \Big)\zeta(q_{ij})
+\frac{1}{2}\sum\limits_{i\neq j}^N\Big( \al_i^4\al_j^2+\al_i^2\al_j^4-c(\al_i^2-\al_j^2)^2 \Big)\wp(q_{ij})
 }
 \end{array}
 \eq
and
  \beq\label{e59}
  \begin{array}{c}
  \displaystyle{
 {\dot q}_i=-2\al_i^2(p_i+\kappa)\,,
 }
 \end{array}
 \eq
  \beq\label{e60}
  \begin{array}{c}
  \displaystyle{
 {\dot p}_i=-k\p_x\Big( p_i^2+2\kappa p_i+k^2\frac{\al_{i,xx}}{\al_i} \Big)
 -2\sum\limits_{j:j\neq i}^N
 \Big( \al_j^6\wp'(q_{ij})-6\al_j^3\al_{j,x}\wp(q_{ij})-k^2\p_x^2(\al_j^2)\zeta(q_{ij}) \Big)\,.
  }
 \end{array}
 \eq
Notice that the subspace $q_1(t,x)+...+q_N(t,x)={\rm const}$ is invariant with respect to dynamics.
%, that is
%we may restrict ourselves to it.
For simplicity we choose the constant to be zero:
  \beq\label{e601}
  \begin{array}{c}
  \displaystyle{
 \sum\limits_{j=1}^N q_j=0\,.
  }
 \end{array}
 \eq
 At the same time the sum of the momenta fields $p_j(t,x)$ is not a constant.

\paragraph{Zakharov-Shabat equation.}  It was shown in \cite{Krich22} that
the equations of motion are represented in the form of the Zakharov-Shabat equations (\ref{e09}).
Introduce the $U$-$V$ pair:
  \beq\label{e61}
  \begin{array}{c}
  \displaystyle{
 U^{\hbox{\tiny{2dCM}}}_{ij}(z)=\delta_{ij}
 \Big( p_i+\al_i^2E_1(z)-k\frac{\al_{i,x}}{\al_i} \Big)
 +(1-\delta_{ij})\phi(z,q_i-q_j)\al_j^2\,,
  }
 \end{array}
 \eq
  \beq\label{e62}
  \begin{array}{c}
  \displaystyle{
 V^{\hbox{\tiny{2dCM}}}_{ij}(z)=\delta_{ij}
 \Big( q_{i,t}E_1(z)-Nc\al_i^2\wp(z)-m_i^0-\frac{\al_{i,t}}{\al_i} \Big)+
   }
  \\ \ \\
  \displaystyle{
 +(1-\delta_{ij})\Big(Ncf(z,q_i-q_j)-NcE_1(z)\phi(z,q_i-q_j)-m_{ij}\phi(z,q_i-q_j)\Big)\al_j^2\,,
  }
 \end{array}
 \eq
 where
  \beq\label{e63}
  \begin{array}{c}
  \displaystyle{
 m_i^0=p_i^2+2\kappa p_i+k^2\frac{\al_{i,xx}}{\al_i}-
 \sum\limits_{j:j\neq i}^N\Big(
 (2\al_j^4+\al_i^2\al_j^2)\wp(q_i-q_j)+4k\al_j\al_{j,x}\zeta(q_i-q_j)
 \Big)\,,
 %\quad i=1,...,N\,,
  }
 \end{array}
 \eq
 and
  \beq\label{e64}
  \begin{array}{c}
  \displaystyle{
 m_{ij}=p_i+p_j+2\kappa-k\frac{\al_{i,x}}{\al_i}+k\frac{\al_{j,x}}{\al_j}+
 \sum\limits_{l:l\neq i,j}^N \al_l^2\eta(q_i,q_l,q_j)\,,\quad i\neq j
  }
 \end{array}
 \eq
 with
  \beq\label{e65}
  \begin{array}{c}
  \displaystyle{
 \eta(z_1,z_2,z_3)=\zeta(z_1-z_2)+\zeta(z_2-z_3)+\zeta(z_3-z_1)\stackrel{(\ref{a05})}{=}
  }
  \\ \ \\
    \displaystyle{
=
 E_1(z_1-z_2)+E_1(z_2-z_3)+E_1(z_3-z_1)\,.
  }
 \end{array}
 \eq
The Lax pair (\ref{e61})-(\ref{e62}) differs from the one derived in \cite{Krich22} by some simple redefinitions and diagonal gauge
transformation. In \cite{Krich22} the Zakharov-Shabat equations appeared as a compatibility condition
for a pair of linear problems arising from considering special solutions of KP equations. That is, (\ref{e09})
holds true by construction.
In our paper we do not discuss relation to KP hierarchy. For this reason let us make a comment on the
verification of the Zakharov-Shabat equation. In contrast to the Lax equation in the finite-dimensional case
(\ref{e21})-(\ref{e22}) this is a quite
tedious calculation. From the diagonal part or (\ref{e09}) one gets
  \beq\label{e66}
  \begin{array}{c}
  \displaystyle{
 {\dot p}_i+k\p_x m_i^0+
  }
  \\ \ \\
  \displaystyle{
 +Nc\al_i^2\sum\limits^N_{l:l\neq i}\al_l^2\wp'(q_{il})
 -2\al_i^2\sum\limits^N_{l:l\neq i}\al_l^2\wp(q_{il})
 \Big( -k\frac{\al_{i,x}}{\al_i}+k\frac{\al_{l,x}}{\al_l}+
 \sum\limits_{j:j\neq i,l}^N \al_j^2\eta(q_i,q_j,q_l) \Big)=0\,.
  }
 \end{array}
 \eq
In order to transform it to the form (\ref{e60}) one should use (\ref{a12}).

When verifying the off-diagonal part of the Zakharov-Shabat equation one
needs the following identity. For any $1 \leq i\neq j\leq N$
  \beq\label{e67}
  \begin{array}{c}
  \displaystyle{
 \sum\limits^N_{\substack{ s\neq l  \\ s\neq i,j;\, l\neq i,j }}
 a_la_s\Big( \eta(q_j,q_l,q_s)+\eta(q_i,q_l,q_s) \Big)\eta(q_i,q_l,q_j)=
 \sum\limits_{l\neq i,j}^N a_l\Big(\wp(q_{il})-\wp(q_{jl})\Big)\sum\limits_{s\neq i,j,l}^N a_s\,,
% \quad\forall i\neq j\,,
  }
 \end{array}
 \eq
 where $a_1,...,a_N$ are arbitrary complex variables. It can be proved by comparing the periodic properties
 and poles at both sides. For the calculations in the off-diagonal part of the Zakharov-Shabat equation
 we put $a_i=\al_i^2$. Due to (\ref{e601}) $\sum_{i=1}^N\al_i^2=Nc$, and we have
  \beq\label{e68}
  \begin{array}{c}
  \displaystyle{
 \sum\limits^N_{\substack{ s\neq l  \\ s\neq i,j;\, l\neq i,j }}
 \al_l^2\al_s^2\Big( \eta(q_j,q_l,q_s)+\eta(q_i,q_l,q_s) \Big)\eta(q_i,q_l,q_j)=
   }
  \\
  \displaystyle{
 =\sum\limits_{l\neq i,j}^N \al_l^2\Big( Nc-\al_i^2-\al_j^2-\al_l^2 \Big)\Big(\wp(q_{il})-\wp(q_{jl})\Big)\,,
 \quad\forall i\neq j\,.
  }
 \end{array}
 \eq
 The rest of calculations are performed straightforwardly with the help of (\ref{a04}), (\ref{a061}) and (\ref{a08})-(\ref{a11}).

\paragraph{Classical $r$-matrix.} Main statement is as follows.

%\noindent {\bf Theorem.}
\begin{theor}
The $U$-matrix (\ref{e61})
satisfies the Maillet $r$-matrix structure (\ref{e11}) with the $r$-matrix for the finite-dimensional Calogero-Moser model (\ref{e26}), where positions of particles $q_i$ are replaced with the fields $q_i(x)$:
\beq\label{e69}
  \begin{array}{c}
  \displaystyle{
 {\bf r}^{\hbox{\tiny{2dCM}}}_{12}(z,w|x) =
 (E_1(z-w)+E_1(w))\,\sum\limits_{i=1}^NE_{ii}\otimes E_{ii}+
  }
  \\
   \displaystyle{
+
 \sum\limits^N_{i\neq j}\phi(z-w,q_i(x)-q_j(x))\,E_{ij}\otimes E_{ji}
 -\sum\limits^N_{i\neq j}\phi(-w,q_i(x)-q_j(x))\,E_{ii}\otimes E_{ji}\,.
 }
 \end{array}
 \eq
\end{theor}
\begin{proof}
  In order to prove this statement we use the standard relations for the delta-function and
 its derivatives:
  \beq\label{e70}
  \begin{array}{c}
  \displaystyle{
 f(y)\delta'(x-y)=f(x)\delta'(x-y)+f'(x)\delta(x-y)\,,
  }
 \end{array}
 \eq
  \beq\label{e71}
  \begin{array}{c}
  \displaystyle{
 f(y)\delta''(x-y)=f(x)\delta''(x-y)+2f'(x)\delta'(x-y)+f''(x)\delta(x-y)
  }
 \end{array}
 \eq
or
  \beq\label{e72}
  \begin{array}{c}
  \displaystyle{
 (f(x)-f(y))\delta''(x-y)=-(f'(x)+f'(y))\delta'(x-y)\,.
  }
 \end{array}
 \eq
 Using (\ref{e72}) one can verify that the following map is canonical:
  \beq\label{e73}
  \begin{array}{c}
  \displaystyle{
 p_i(x)\rightarrow p_i(x)+b\frac{k^2q_{i,xx}}{kq_{i,x}+c}=
  p_i(x)+\frac{bk}{2}\frac{\al_{i,x}}{\al}\,,
  }
  \\ \ \\
    \displaystyle{
 q_i(x)\rightarrow q_i(x)\,,
  }
 \end{array}
 \eq
 where $b$ and $c$ are arbitrary constants. This is needed for the computation of the
 left-hand side of (\ref{e11}). Namely, we have:
  \beq\label{e74}
  \begin{array}{c}
  \displaystyle{
\{U_1(z,x),U_2(w,y)\}=-k\sum\limits_{i=1}^N E_{ii}\otimes E_{ii}\Big( E_1(z)+E_1(w) \Big)\delta'(x-y)+
  }
  \\
  \displaystyle{
+\sum\limits_{i\neq j}^N E_{ii}\otimes E_{ij}\, f(w,q_{ij}(x))\al_j^2(x)\delta(x-y)
-\sum\limits_{i\neq j}^N E_{ij}\otimes E_{ii}\, f(z,q_{ij}(x))\al_j^2(x)\delta(x-y)+
 }
   \\
  \displaystyle{
+\sum\limits_{i\neq j}^N E_{ij}\otimes E_{jj}
\Big(f(z,q_{ij}(x))\al_j^2(x)\delta(x-y)-\phi(z,q_{ij}(x))k\delta'(x-y)\Big)-
 }
    \\
  \displaystyle{
-\sum\limits_{i\neq j}^N E_{ii}\otimes E_{ji}
\Big(f(w,q_{ji}(x))\al_i^2(x)\delta(x-y)+\phi(w,q_{ji}(y))k\delta'(x-y)\Big)\,.
 }
 \end{array}
 \eq
 The expression in the last line can be transformed using (\ref{e70}):
  \beq\label{e75}
  \begin{array}{c}
  \displaystyle{
 f(w,q_{ji}(x))\al_i^2(x)\delta(x-y)+\phi(w,q_{ji}(y))k\delta'(x-y)=
  }
    \\ \ \\
  \displaystyle{
 =f(w,q_{ji}(x))\al_i^2(x)\delta(x-y)+\phi(w,q_{ji}(x))k\delta'(x-y)
 +k(q_{j,x}-q_{i,x})f(w,q_{ji}(x))\delta(x-y)=
  }
   \\ \ \\
     \displaystyle{
 = f(w,q_{ji}(x))\al_j^2(x)\delta(x-y)+\phi(w,q_{ji}(x))k\delta'(x-y)\,.
  }
 \end{array}
 \eq
Next, consider the right-hand side of (\ref{e11}). The commutator terms are evaluated
using (\ref{a07})-(\ref{a08}):
  \beq\label{e76}
  \begin{array}{c}
  \displaystyle{
[U_1(z,x),{\bf r}^{\hbox{\tiny{2dCM}}}_{12}(z,w|x)]-[U_2(w,y),{\bf r}^{\hbox{\tiny{2dCM}}}_{21}(w,z|x)]=
  }
  \\ \ \\
  \displaystyle{
+\sum\limits_{i\neq j}^N E_{ii}\otimes E_{ij}\, f(w,q_{ij}(x))\al_j^2(x)
-\sum\limits_{i\neq j}^N E_{ij}\otimes E_{ii}\, f(z,q_{ij}(x))\al_j^2(x)+
 }
   \\
  \displaystyle{
+\sum\limits_{i\neq j}^N E_{ij}\otimes E_{jj}\, f(z,q_{ij}(x))\al_j^2(x)
-\sum\limits_{i\neq j}^N E_{ii}\otimes E_{ji}\, f(w,q_{ji}(x))\al_i^2(x)+
 }
    \\
  \displaystyle{
+\sum\limits_{i\neq j}^N E_{ij}\otimes E_{ji}\, f(z-w,q_{ij}(x))(\al_i^2(x)-\al_j^2(x))\,.
 }
 \end{array}
 \eq
The derivative term gives
  \beq\label{e77}
  \begin{array}{c}
  \displaystyle{
k\p_x{\bf r}^{\hbox{\tiny{2dCM}}}_{12}(z,w|x)=
\sum\limits_{i\neq j}^N E_{ii}\otimes E_{ji}\,
f(w,q_{ji}(x))k(q_{j,x}-q_{i,x})+
  }
  \\
  \displaystyle{
+\sum\limits_{i\neq j}^N E_{ij}\otimes E_{ji}\, f(z-w,q_{ij}(x))k(q_{i,x}-q_{j,x})\,.
 }
 \end{array}
 \eq
 Finally, the non-ultralocal term yields
  \beq\label{e78}
  \begin{array}{c}
  \displaystyle{
 \Big({\bf r}^{\hbox{\tiny{2dCM}}}_{12}(z,w|x)+{\bf r}^{\hbox{\tiny{2dCM}}}_{21}(w,z|x)\Big)k\delta'(x-y)=
 \sum\limits_{i=1}^N E_{ii}\otimes E_{ii}\Big( E_1(z)+E_1(w) \Big)k\delta'(x-y)+
 }
 \\
   \displaystyle{
+\sum\limits_{i\neq j}^N E_{ij}\otimes E_{jj}\, \phi(z,q_{ij}(x))k\delta'(x-y)
+\sum\limits_{i\neq j}^N E_{ii}\otimes E_{ji}\, \phi(w,q_{ji}(x))k\delta'(x-y)\,.
  }
 \end{array}
 \eq
 By summing up all the inputs the statement is proved.
 \end{proof}

 \paragraph{Modified classical $r$-matrix.} In the end of the paragraph consider also the $U$-matrix of the field Calogero-Moser model
 in the center of mass frame from the very beginning:
  \beq\label{e781}
  \begin{array}{c}
  \displaystyle{
 {\bar U}^{\hbox{\tiny{2dCM}}}_{ij}(z)=\delta_{ij}
 \Big( p_i+{\bar \al}_i^2E_1(z)-k\frac{{\bar \al}_{i,x}}{{\bar \al}_i} \Big)
 +(1-\delta_{ij})\phi(z,q_i-q_j){\bar \al}_j^2\,,
  }
 \end{array}
 \eq
 where
  \beq\label{e782}
  \begin{array}{c}
  \displaystyle{
{\bar \al}_j^2=k{\bar q}_{j,x}+c
  }
 \end{array}
 \eq
 and
  \beq\label{e783}
  \begin{array}{c}
  \displaystyle{
{\bar q}_{j}=q_j-\frac{1}{N}\sum\limits_{k=1}^N q_k\,.
  }
 \end{array}
 \eq
 Straightforward calculation shows that in this case the classical $r$-matrix entering the Maillet bracket
 is of the form:
  \beq\label{e784}
  \begin{array}{c}
  \displaystyle{
{\bar {\bf r}}^{\hbox{\tiny{2dCM}}}_{12}(z,w|x)=
{\bf r}^{\hbox{\tiny{2dCM}}}_{12}(z,w|x)-{\bar l}_2(w,x)\,,\qquad {\bar l}_2(w,x)=1_N\otimes {\bar l}(w,x)\,,
  }
 \end{array}
 \eq
 and
  \beq\label{e785}
  \begin{array}{c}
  \displaystyle{
{\bar l}_{ij}(w,x)=\frac{1}{N}\Big(\delta_{ij}E_1(w)+(1-\delta_{ij})\phi(w,q_i(x)-q_j(x))\Big)\,.
  }
 \end{array}
 \eq

\subsection{1+1 spin and multispin Calogero-Moser model}
%
%\paragraph{Unreduced model.}
The description of 1+1 version of the spin Calogero-Moser model is very
similar to the finite-dimensional case. Introduce a set of fields on a circle $\mS_{ij}(x)$, $i,j=1...N$ --
field analogues of the spin variables $S_{ij}$. The matrix $\mS(x)$ can be considered as an element
of the loop coalgebra $L^*({\rm gl}_N)$, which is dual to the loop algebra $L({\rm gl}_N)$ of ${\mathcal C}^\infty$-maps
${\mathbb S}^1\rightarrow {\rm gl}_N$. The space $L^*({\rm gl}_N)$ is endowed with the linear Poisson
structure:
  \beq\label{e79}
  \begin{array}{c}
  \displaystyle{
   \{\mS_{ij}(x),\mS_{kl}(y)\}=\Big(-\mS_{il}(x)\delta_{kj}+\mS_{kj}(x)\delta_{il}\Big)\delta(x-y)\,.
 }
 \end{array}
 \eq
A coadjoint orbit $\mO$ appears by fixing the Casimir functions of (\ref{e79}).
Together with the fields $p_i(x)$, $q_j(x)$ with the canonical Poisson bracket (\ref{e08})
this provides a description of the unreduced phase space.
%-- the eigenvalues of $\mS(x)$.
%nujen analog (\ref{e42})
General construction was described in \cite{LOZ}. In the elliptic case it leads to
the following $U$-matrix:
  \beq\label{e80}
  \begin{array}{c}
  \displaystyle{
U^{\hbox{\tiny{2dSpin}}}_{ij}(z,x)=
\delta_{ij}(p_i+\mS_{ii}(x)E_1(z))+(1-\delta_{ij})\mS_{ij}(x)\phi(z,q_i(x)-q_j(x))\,,
 }
 \end{array}
 \eq
which is a straightforward field analogue of the finite-dimensional Lax matrix (\ref{e29}).
At the same time, in contrast to (\ref{e36}), the coadjoint action of the loop group provides the
constraint
  \beq\label{e801}
  \begin{array}{c}
  \displaystyle{
\mS_{ii}(x)=k\p_xq_{i}(x)+c\,,\quad i=1,...,N
 }
 \end{array}
 \eq
 or
  \beq\label{e8011}
  \begin{array}{c}
  \displaystyle{
\mS_{ii}(x)=\al_i^2(x)\,,\quad i=1,...,N\,.
 }
 \end{array}
 \eq
The Hamiltonian formulation and derivation of equations of motion is a separate non-trivial task, and we
do not go into it here.

Consider the $r$-matrix (\ref{e40})
\beq\label{e802}
  \begin{array}{c}
  \displaystyle{
 {\bf r}^{\hbox{\tiny{2dSpin}}}_{12}(z,w|x)=
 E_1(z-w)\,\sum\limits_{i=1}^NE_{ii}\otimes E_{ii}+
 \sum\limits^N_{i\neq j}\phi(z-w,q_{i}(x)-q_j(x))\,E_{ij}\otimes E_{ji}\,.
 }
 \end{array}
 \eq
 %
%\noindent {\bf Theorem 2} {\em
\begin{theor}
Then the $U$-matrix (\ref{e80}) and $r$-matrix (\ref{e802})
 satisfy the Maillet ultralocal $r$-matrix structure with the unwanted term
  \beq\label{e803}
  \begin{array}{c}
  \displaystyle{
\{U_1(z,x),U_2(w,y)\}=
  }
  \\ \ \\
  \displaystyle{
\Big(-k\p_x {\bf r}_{12}(z,w|x)
+ [U_1(z,x),{\bf r}_{12}(z,w|x)]-[U_2(w,y),{\bf r}_{21}(w,z|x)]\Big)\delta(x-y)-
  }
%  \\ \ \\
%  \displaystyle{
%  -\Big({\bf r}_{12}(z,w|x)+{\bf r}_{21}(w,z|x)\Big)\delta'(x-y)
% }
 \\ \ \\
  \displaystyle{
  -\sum\limits_{k\neq l}^N E_{kl}\otimes E_{lk} \Big(\mu_k(x)-\mu_l(x)\Big)f(z-w,q_k(x)-q_l(x))\delta(x-y)\,,
 }
 \end{array}
 \eq
where
\beq\label{e804}
  \begin{array}{c}
  \displaystyle{
\mu_i(x)=\mS_{ii}(x)-k\p_xq_{i}(x)\,,\quad i=1,...,N\,.
 }
 \end{array}
 \eq
\end{theor}
\begin{proof}This is the field analogue of the finite-dimensional relation (\ref{e39}). The non-ultralocal term is absent here
since the $r$-matrix (\ref{e802}) is skew-symmetric, see (\ref{e41}). The proof of the relation (\ref{e803})
is simple. Indeed, similarly to the finite-dimensional case
$U$-matrix (\ref{e80}) and $r$-matrix (\ref{e802}) satisfy a direct field generalization of (\ref{e39}) with the factor $\delta(x-y)$ for all the terms. But (\ref{e39}) does not contain the term $\p_x {\bf r}_{12}(z,w|x)$. Its input
in the field case
\beq\label{e805}
  \begin{array}{c}
  \displaystyle{
 k\p_x {\bf r}^{\hbox{\tiny{2dSpin}}}_{12}(z,w)=
 \sum\limits^N_{i\neq j}f(z-w,q_{i}-q_j)(kq_{i,x}-kq_{j,x})\,E_{ij}\otimes E_{ji}
 }
 \end{array}
 \eq
is added to the unwanted term providing $\mu_i$ instead of $\mS_{ii}$ in (\ref{e39}).
\end{proof}

The unwanted term in (\ref{e803}) vanishes when
\beq\label{e806}
  \begin{array}{c}
  \displaystyle{
 \mu_i=c\,,\quad i=1,...,N\,,
 }
 \end{array}
 \eq
that is when the conditions (\ref{e801}) hold true. On shell constraints (\ref{e806})
the relation (\ref{e803}) turns into the Maillet bracket. Similarly to the finite-dimensional case
one should perform Poisson
reduction with respect to the coadjoint action of the Cartan subgroup in $L({\rm GL}_N)$ by
supplying (\ref{e806}) with some gauge fixing conditions.
In this way it is possible to obtain the reduced $r$-matrix, and it can be non-ultralocal. Indeed,
the reduction (\ref{e43}) makes the $r$-matrix non skew-symmetric, and this leads to the non-ultralocal terms.

\paragraph{Special case: reduction to the spinless model.}
Consider the case when $\mS$ is a rank one matrix, that is
\beq\label{e807}
  \begin{array}{c}
  \displaystyle{
 \mS_{ij}(x)=\xi_i(x)\eta_j(x)\,.
 }
 \end{array}
 \eq
Similarly to (\ref{e44})-(\ref{e45}) the bracket (\ref{e79})
is parameterized by the canonical Poisson bracket
\beq\label{e808}
  \begin{array}{c}
  \displaystyle{
 \{\xi_i(x),\eta_j(y)\}=\delta_{ij}\delta(x-y)\,.
 }
 \end{array}
 \eq
Next, we perform the same calculation as it was described in (\ref{e46})-(\ref{e48}). Namely, perform the gauge
transformation
\beq\label{e809}
  \begin{array}{c}
  \displaystyle{
 U^{\hbox{\tiny{2dSpin}}}(z,x)\rightarrow
 g^{-1} U^{\hbox{\tiny{2dSpin}}}(z,x) g-kg^{-1}\p_x g\,,\quad g={\rm diag}(\xi_1(x),...,\xi_N(x))
 }
 \end{array}
 \eq
The resultant $U$-matrix has the following entries
\beq\label{e810}
  \begin{array}{c}
  \displaystyle{
 \delta_{ij}
 \Big( p_i+\mS_{ii}E_1(z)-k\frac{\xi_{i,x}}{\xi_i} \Big)
 +(1-\delta_{ij})\phi(z,q_i-q_j)\mS_{jj}
 }
 \end{array}
 \eq
By imposing constraints (\ref{e8011}) and choosing gauge fixing condition as $\xi_i=\al_i$
one gets exactly the $U$-matrix for the spinless case (\ref{e61}). It is described
by the $r$-matrix (\ref{e69}), which provides the non-ultralocal terms.

\paragraph{1+1 multispin models.}
Following \cite{LOZ}  define the $U$-matrix for the multi-pole (or multi-spin) model:
  \beq\label{e81}
  \begin{array}{c}
  \displaystyle{
U^{\hbox{\tiny{mult-spin}}}_{ij}(z,x)=
\delta_{ij}(p_i+\sum\limits_{a=1}^M\mS^a_{ii}(x)E_1(z-z_a))+
(1-\delta_{ij})\sum\limits_{a=1}^M\mS^a_{ij}(x)\phi(z-z_a,q_i(x)-q_j(x))\,.
 }
 \end{array}
 \eq
Its form is exactly the same as in the finite-dimensional case (\ref{e49}).
Description of this model is similar to the spin case but the phase space is larger.
It contains $M$ components with the Poisson structure
  \beq\label{e82}
  \begin{array}{c}
  \displaystyle{
   \{\mS^a_{ij}(x),\mS^b_{kl}(y)\}=\delta^{ab}
   \Big(-\mS^a_{il}(x)\delta_{kj}+\mS^a_{kj}(x)\delta_{il}\Big)\delta(x-y)\,,\quad
   a,b=1,...,M\,.
 }
 \end{array}
 \eq
The Cartan subgroup acts on all the components simultaneously providing the constraint
  \beq\label{e83}
  \begin{array}{c}
  \displaystyle{
\sum\limits_{a=1}^M\mS^a_{ii}(x)=\al_i^2(x)\,,\quad i=1,...,N\,,
 }
 \end{array}
 \eq
 which is a field analogue of the constraint (\ref{e53}).

The $r$-matrix structure on the unreduced phase space is the same as in the spin case.
More precisely, it is given by (\ref{e803}) but with
\beq\label{e84}
  \begin{array}{c}
  \displaystyle{
\mu_i(x)=\sum\limits_{a=1}^M\mS^a_{ii}(x)-k\p_xq_{i}(x)\,,\quad i=1,...,N\,.
 }
 \end{array}
 \eq
 %

%%%%%%%%%%%%%%%%%%%%%%%%%%%%%%%%%%%%%%%%%%%%%%%%%%%%%%%%%%%%%%%%%%%%%%%%%%%%%%%%%%%%%%%%%%%
%%%%%%%%%%%%%%%%%%%%%%%%%%%%%%%%%%%%%%%%%%%%%%%%%%%%%%%%%%%%%%%%%%%%%%%%%%%%%%%%%%%%%%%%%%%
\section{Continuous version of the classical IRF-Vertex relation}\label{sec4}
\setcounter{equation}{0}

Up till now we considered the models related to certain dynamical $r$-matrices.
There is also a class of 1+1 models related to non-dynamical elliptic Belavin-Drinfeld $r$-matrix \cite{BD}.
The latter is defined as
 \beq\label{e85}
 \begin{array}{c}
  \displaystyle{
 r^{\hbox{\tiny{BD}}}_{12}(z,w)=E_1(z-w)\frac{1_N\otimes 1_N}{N}+\sum\limits_{\substack{a\in\,
 \mZ_{ N}\times\mZ_{ N} \\ a\neq(0,0)}} \frac{T_a\otimes T_{-a}}{N}
 \exp (2\pi\imath\,\frac{a_2(z-w)}{N})\,\phi(z-w,\frac{a_1+a_2\tau}{N})\,,
 }
 \end{array}
 \eq
being written in a special matrix basis (in $\Mat$)
 \beq\label{e86}
 \begin{array}{c}
  \displaystyle{
T_a=T_{a_1 a_2}=\exp\left(\frac{\pi\imath}{{ N}}\,a_1
 a_2\right)Q_1^{a_1}Q_2^{a_2}\in\Mat\,,\quad a=(a_1,a_2)\in\mZ_{ N}\times\mZ_{ N}\,,
 }
 \end{array}
 \eq
 defined in terms of the pair of matrices
 \beq\label{e87}
 \begin{array}{c}
  \displaystyle{
(Q_1)_{kl}=\delta_{kl}\exp(\frac{2\pi\imath}{{ N}}k)\,,\qquad
(Q_2)_{kl}=\delta_{k-l+1=0\,{\hbox{\tiny{mod}}}\,
 { N}}\,.
 }
 \end{array}
 \eq
The $r$-matrix (\ref{e85}) is non-dynamical, that is it is independent of dynamical variables.
In particular, it means that this $r$-matrix is exactly the same in the field case:
 \beq\label{e851}
 \begin{array}{c}
  \displaystyle{
 {\bf r}^{\hbox{\tiny{BD}}}_{12}(z,w|x)=r^{\hbox{\tiny{BD}}}_{12}(z,w)
 }
 \end{array}
 \eq
 since
 \beq\label{e852}
 \begin{array}{c}
  \displaystyle{
 \p_x {\bf r}^{\hbox{\tiny{BD}}}_{12}(z,w|x)=0\,.
 }
 \end{array}
 \eq
 Moreover, $r$-matrix (\ref{e85}) is skew-symmetric\footnote{A review of its properties can be found for example
in \cite{ZZ}.}: $r^{\hbox{\tiny{BD}}}_{12}(z,w)=-r^{\hbox{\tiny{BD}}}_{21}(w,z)$, so that only the commutator terms survive in the Maillet bracket (\ref{e11}).
The corresponding $U$-matrix is the one of the Sklyanin type, i.e.
 \beq\label{e88}
 \begin{array}{c}
  \displaystyle{
U^{\hbox{\tiny{Skl}}}(z,x)=\sum\limits_{\substack{a\in\,
 \mZ_{ N}\times\mZ_{ N} \\ a\neq(0,0)}} T_a\mS_a(x) \exp (2\pi\imath\,\frac{a_2z}{N})\,\phi(z,\frac{a_1+a_2\tau}{N})\,,
 }
 \end{array}
 \eq
where $\mS_a(x)$ are dynamical field variables. From the above arguments it follows that
  \beq\label{e8033}
  \begin{array}{c}
  \displaystyle{
\{U^{\hbox{\tiny{Skl}}}_1(z,x),U^{\hbox{\tiny{Skl}}}_2(w,y)\}=
  }
  \\ \ \\
  \displaystyle{
=\Big(
 [U^{\hbox{\tiny{Skl}}}_1(z,x),{\bf r}^{\hbox{\tiny{BD}}}_{12}(z,w|x)]-[U^{\hbox{\tiny{Skl}}}_2(w,y),{\bf r}^{\hbox{\tiny{BD}}}_{21}(w,z|x)]\Big)\delta(x-y)=
  }
  \\ \ \\
  \displaystyle{
  =
 [U^{\hbox{\tiny{Skl}}}_1(z,x)+U^{\hbox{\tiny{Skl}}}_2(w,x),{\bf r}^{\hbox{\tiny{BD}}}_{12}(z,w|x)]\delta(x-y)\,.
  }
 \end{array}
 \eq
 This relation is equivalent to the Poisson brackets (\ref{e79}) for the components of matrix $\mS$.
In the case $N=2$ the $U$-matrix (\ref{e88}) is the one for the Landau-Lifshitz model \cite{Skl},
and three fields $\mS_a(x)$ are the components of the magnetization vector entering (\ref{e055}).
More general cases including the multi-pole generalizations and trigonometric and rational degenerations of (\ref{e88})
were studied in \cite{Z11,AtZ2}.

An existence of a certain gauge transformation relating 1+1 models of the Sklyanin-Landau-Lifshitz type on one hand
and the models of 1+1 Calogero-Moser field theory on the other hand was mentioned in
\cite{LOZ} and explicit change of variables was derived in \cite{AtZ1,AtZ3}.
At the same time, a similar relation between dynamical and non-dynamical $r$-matrices is widely known.
It is the IRF-Vertex correspondence \cite{Baxter2}. See a brief review in the Appendix B.
The relation between $r$-matrices is written in terms of
intertwining matrix
  \beq\label{e89}
  \begin{array}{l}
  \displaystyle{
 g(z,q)=\Xi(z,q)\left(d^{0}\right)^{-1}
 }
 \end{array}
 \eq
 with
 \beq\label{e90}
 \begin{array}{c}
  \displaystyle{
\Xi_{ij}(z,q)=
 \vth\left[  \begin{array}{c}
 \frac12-\frac{i}{N} \\ \frac N2
 \end{array} \right] \left(z-Nq_j+\sum\limits_{m=1}^N
 q_m\left.\right|N\tau\right)\,,
 }
 \end{array}
 \eq
and the diagonal matrix
 \beq\label{e91}
 \begin{array}{c}
  \displaystyle{
d^0_{ij}(z,q)=\delta_{ij}d^0_{j}=\delta_{ij}
 {\prod\limits_{k:k\neq j}^N\vth(q_j-q_k)}\,,
 }
 \end{array}
 \eq
 where $q_i$ are parameters in the IRF model. The same matrix is used to describe gauge equivalence
 between the finite-dimensional integrable systems, where the vertex type models are integrable tops,
 while the systems of Calogero-Ruijsenaars family play the role of IRF models, and $q_i$
are positions of particles.

Notice that the formula (\ref{e42}) describing the gauge transformation of $r$-matrix
is also valid in the field case.
%\footnote{We changed $g$ to $G^{-1}$ in (\ref{e92}) with respect to (\ref{e42}).}.
Consider a gauge transformation with ''ultra-local brackets'', that is suppose that
the Poisson brackets for a matrix $G(z,x)\in\Mat$ of gauge transformation have the form
 \beq\label{e9101}
 \begin{array}{c}
  \displaystyle{
\{G_1(z,x),U_2(w,y)\}=\langle\{G_1,U_2\}\rangle(z,w,x)\delta(x-y)\,,
 }
 \\ \ \\
   \displaystyle{
\{G_1(z,x),G_2(w,y)\}=\langle\{G_1,G_2\}\rangle(z,w,x)\delta(x-y)\,,
 }
 \end{array}
 \eq
 where the expressions $\langle\{G_1,U_2\}\rangle(z,w,x)$ and $\langle\{G_1,G_2\}\rangle(z,w,x)$ are
 some $\Mat^{\otimes 2}$-valued functions of the field variables obtained by direct calculation
 from a given Poisson structure.
 \begin{lemma}
 Suppose $U\in\Mat$ satisfies the Maillet bracket (\ref{e11}) with some
 $r$-matrix ${\bf r}_{12}(z,w,x)$. Consider a gauge transformation
 \beq\label{e9102}
 \begin{array}{c}
  \displaystyle{
U(z,x)\rightarrow G(z,x)U(z,x)G^{-1}(z,x)+k\p_x G(z,x)G^{-1}(z,x)
 }
 \end{array}
 \eq
 with a matrix of gauge transformation satisfying the property of ultra-local brackets (\ref{e9101}).
 Then the gauge transformed $U$-matrix $U'(z,x)=G(z,x)U(z,x)G^{-1}(z,x)+k\p_x G(z,x)G^{-1}(z,x)$ again
 satisfies the Maillet bracket (\ref{e11}) with the gauge transformed $r$-matrix
\beq\label{e92}
  \begin{array}{c}
  \displaystyle{
 {\bf r}_{12}(z,w|x)\longrightarrow
 G_1(z,x)G_2(w,x)\Big(
 {\bf r}_{12}(z,w|x)-G_1^{-1}(z,x)\langle\{G_1,U_2\}\rangle(z,w,x)-
 }
 \end{array}
 \eq
 $$
 -\frac12\,\Big[G_1^{-1}(z,x)G_2^{-1}(w,x)\langle\{G_1,G_2\}\rangle(z,w,x),U_2(w,x)\Big]
 \Big)G^{-1}_1(z,x)G^{-1}_2(w,x)\,.
 $$
That is, the Maillet bracket (\ref{e11}) keeps its form under the gauge transformation of $U$-matrices
(\ref{e111}) and the transformation of $r$-matrix (\ref{e92}).
%Here we mean any $U$-matrix and the corresponding $r$-matrix satisfying (\ref{e11}).
This is a direct generalization of the gauge transformation (\ref{e42})
 in the finite-dimensional mechanics.
\end{lemma}
\begin{proof}
Let us represent the Maillet bracket (\ref{e11}) in the finite-dimensional form
(\ref{e04}):
  \beq\label{e93}
  \begin{array}{c}
  \displaystyle{
\{\nabla U_1(z,x),\nabla U_2(w,y)\}=
}
\\ \ \\
  \displaystyle{
 =\Big([\nabla U_1(z,x),{\bf r}_{12}(z,w|x)\delta(x-y)]-[\nabla U_2(w,y),{\bf r}_{21}(w,z|x)\delta(x-y)]\Big)\,,
  }
 \end{array}
 \eq
where $\nabla U(z,x)=-k\p_x+U(z,x)$. Indeed,
  \beq\label{e931}
  \begin{array}{c}
  \displaystyle{
[\nabla U_1(z,x),{\bf r}_{12}(z,w|x)\delta(x-y)]=[U_1(z,x),{\bf r}_{12}(z,w|x)]\delta(x-y)-
  }
  \\ \ \\
  \displaystyle{
-k(\p_x{\bf r}_{12}(z,w|x))\delta(x-y)
-{\bf r}_{12}(z,w|x)k\delta'(x-y)
  }
 \end{array}
 \eq
 and
  \beq\label{e932}
  \begin{array}{c}
  \displaystyle{
[\nabla U_2(w,y),{\bf r}_{21}(w,z|x)\delta(x-y)]=[U_2(w,y),{\bf r}_{21}(w,z|x)]\delta(x-y)
+{\bf r}_{21}(w,z|x)k\delta'(x-y)\,.
  }
 \end{array}
 \eq
 The difference of (\ref{e931}) and (\ref{e932}) yields (\ref{e11}).

Thus we have $\nabla U(z,x)\rightarrow G(z,x)\nabla U(z,x)G^{-1}(z,x)$, and the gauge transformation formulae are the same as in the
finite-dimensional case.
In fact, with the conditions (\ref{e9101}) the gauge transformation can be also rewritten as
\beq\label{e9321}
  \begin{array}{c}
  \displaystyle{
 {\bf r}_{12}(z,w|x)\delta(x-y)\longrightarrow
 G_1(z,x)G_2(w,x)\Big(
 {\bf r}_{12}(z,w|x)\delta(x-y)-
 }
 \end{array}
 \eq
 $$
 -G_1^{-1}(z,x)\{G_1(z,x),U_2(w,y)\}-
 $$
 $$
-\frac12\,\Big[G_1^{-1}(z,x)G_2^{-1}(w,x)\{G_1(z,x),G_2(w,y)\},U_2(w,x)\Big]
 \Big)G^{-1}_1(z,x)G^{-1}_2(w,x)\,.
 $$
\end{proof}

Consider the special gauge transformation
\beq\label{e933}
  \begin{array}{c}
  \displaystyle{
{\bar G}(z,x)=\frac{1}{[\det g(z,q(x))]^{1/N}}\,g(z,q(x))\,,
 }
 \end{array}
 \eq
where $g(z,q(x))$ is the intertwining matrix (\ref{e89})-(\ref{e91}) from the IRF-Vertex correspondence
with the dynamical parameters $q_1(x),...,q_N(x)$. Let us formulate the continuous classical version of the IRF-Vertex relation.
\begin{theor}
 The following relation holds true:
\beq\label{e94}
  \begin{array}{c}
  \displaystyle{
 {\bf r}^{\hbox{\tiny{BD}}}_{12}(z,w|x)={r}^{\hbox{\tiny{BD}}}_{12}(z,w)= {\bar G}_1(z,x){\bar G}_2(w,x)\Big(
 {\bar {\bf r}}^{\hbox{\tiny{2dCM}}}_{12}(z,w|x)-
 }
 \\ \ \\
  \displaystyle{
-{\bar G}_1^{-1}(z,x)\langle\{{\bar G}_1,{\bar U}^{\hbox{\tiny{2dCM}}}_2\}\rangle(z,w,x)\Big){\bar G}^{-1}_1(z,x){\bar G}^{-1}_2(w,x)
 }
 \end{array}
 \eq
 %S
 where $ {\bf r}^{\hbox{\tiny{BD}}}_{12}(z,w|x)$ is the Belavin-Drinfeld $r$-matrix,
 ${\bar {\bf r}}^{\hbox{\tiny{2dCM}}}_{12}(z,w|x)$ is the $r$-matrix (\ref{e784}),
 the gauge transformation ${\bar G}(z,x)$ is given by (\ref{e933}) and
 the matrix ${\bar U}^{\hbox{\tiny{2dCM}}}$ is the one from (\ref{e781}).
 \end{theor}
In fact, ${\bar U}^{\hbox{\tiny{2dCM}}}(w,x)$ can be replaced with $P={\rm  diag}(p_1(x),...,p_N(x))$
since only this part of ${\bar U}^{\hbox{\tiny{2dCM}}}(w,x)$ provides a non-trivial input into the bracket.
The upper statement is a straightforward field generalization of the relation (\ref{e929})
in the finite-dimensional case.

Summarizing, we argued that the well-know relation between dynamical and non-dynamical $r$-matrices
remains the same at the level of field theories, and the corresponding gauge transformation
is exactly the intertwining matrix coming from the IRF-Vertex correspondence in quantum statistical
models.

%ssylka na Classical Yang–Baxter Equation, Lagrangian Multiforms and Ultralocal Integrable Hierarchies

%Yang-Baxter deformations of the
%Principal Chiral Model plus Wess-Zumino term
%B. Hoare1 and S. Lacroix2

%%%%%%%%%%%%%%%%%%%%%%%%%%%%%%%%%%%%%%%%%%%%%%%%%%%%%%%%%%%%%%%%%%%%%%%%%%%%%%%%%%%%%%%%%%%
%%%%%%%%%%%%%%%%%%%%%%%%%%%%%%%%%%%%%%%%%%%%%%%%%%%%%%%%%%%%%%%%%%%%%%%%%%%%%%%%%%%%%%%%%%%
%\section{1+1 analogue for the general spin model}\label{sec5}
%\setcounter{equation}{0}

%%%%%%%%%%%%%%%%%%%%%%%%%%%%%%%%%%%%%%%%%%%%%%%%%%%%%%%%%%%%%%%%%%%%%%%%%%%%%%%%%%%%%%%%%%%
%%%%%%%%%%%%%%%%%%%%%%%%%%%%%%%%%%%%%%%%%%%%%%%%%%%%%%%%%%%%%%%%%%%%%%%%%%%%%%%%%%%%%%%%%%%
\section{Appendix A: elliptic functions}\label{secA}
\def\theequation{A.\arabic{equation}}
\setcounter{equation}{0}

We use the Kronecker elliptic function \cite{Weil}:
 \beq\label{a01}
  \begin{array}{l}
  \displaystyle{
 \phi(z,u)=\frac{\vth'(0)\vth(z+u)}{\vth(z)\vth(u)}=\phi(u,z)\,,
  \qquad \phi(-z, -u) = -\phi(z, u)\,,
 }
 \end{array}
 \eq
where $\vth(z)$ is the first Jacobi theta-function,  or  using the Riemann notations it is
 \beq\label{a02}
 \begin{array}{c}
  \displaystyle{
\vth(z)=\vth(z,\tau)\equiv-\theta{\left[\begin{array}{c}
1/2\\
1/2
\end{array}
\right]}(z|\, \tau )\,,
 }
 \end{array}
 \eq
\beq\label{a03}
 \begin{array}{c}
  \displaystyle{
\theta{\left[\begin{array}{c}
a\\
b
\end{array}
\right]}(z|\, \tau ) =\sum_{j\in \z}
\exp\left(2\pi\imath(j+a)^2\frac\tau2+2\pi\imath(j+a)(z+b)\right)\,,\quad {\rm Im}(\tau)>0\,,
}
 \end{array}
 \eq
where $\tau$ is the elliptic moduli.
The derivative $f(z,u) = \partial_u \vf(z,u)$ is given by
\beq\label{a04}
\begin{array}{c} \displaystyle{

    f(z, u) = \phi(z, u)(E_1(z + u) - E_1(u)), \qquad f(-z, -u) = f(z, u)
}\end{array}\eq
in terms of the first
  Eisenstein function:
\beq\label{a05}
\begin{array}{c} \displaystyle{
    E_1(z)=\frac{\vth'(z)}{\vth(z)}=\zeta(z)+\frac{z}{3}\frac{\vth'''(0)}{\vth'(0)}\,,
    \quad
    E_2(z) = - \partial_z E_1(z) = \wp(z) - \frac{\vartheta'''(0) }{3\vartheta'(0)}\,,
}\end{array}\eq
\beq\label{a06}
\begin{array}{c}
 \displaystyle{
    E_1(- z) = -E_1(z)\,, \quad E_2(-z) = E_2(z)\,,
}\end{array}
\eq
where $\wp(z)$ and $\zeta(z)$ are the Weierstrass functions. The second order derivative
$f'(z,u) = \partial^2_u \phi(z,u)$ is
\beq\label{a061}
\begin{array}{c}
 \displaystyle{
    f'(z,u)=\phi(z,u)\Big(\wp(z)-E_1^2(z)+2\wp(u)-2E_1(z)E_1(u)+2E_1(z+u)E_1(z)\Big)=
 }
 \\
 \displaystyle{
   =2\Big(\wp(u)-\rho(z)\Big)\phi(z,u)+2E_1(z)f(z,u) \,,
}\end{array}
\eq
where
\beq\label{a062}
\begin{array}{c}
\displaystyle{
\rho (z) = \frac{E^2_1(z) - \wp(z)}{2}\,.
}
\end{array}
\eq
 The defined above functions satisfy the widely known addition formulae:
\beq\label{a07}
  \begin{array}{c}
  \displaystyle{
  \phi(z_1, u_1) \phi(z_2, u_2) = \phi(z_1, u_1 + u_2) \phi(z_2 - z_1, u_2) + \phi(z_2, u_1 + u_2) \phi(z_1 - z_2, u_1)\,,
 }
 \end{array}
 \eq
\beq\label{a08}
  \begin{array}{c}
  \displaystyle{
 \phi(z,u_1)\phi(z,u_2)=\phi(z,u_1+u_2)\Big(E_1(z)+E_1(u_1)+E_1(u_2)-E_1(z+u_1+u_2)\Big)\,,
 }
 \end{array}
 \eq
\beq\label{a09}
  \begin{array}{c}
  \displaystyle{
  \phi(z, u_1) f(z,u_2)-\phi(z, u_2) f(z,u_1)=\phi(z,u_1+u_2)\Big(\wp(u_1)-\wp(u_2)\Big)\,,
 }
 \end{array}
 \eq
\beq\label{a10}
  \begin{array}{c}
  \displaystyle{
  \phi(z, u) \phi(z, -u) = \wp(z)-\wp(u)=E_2(z)-E_2(u)\,,
 }
 \end{array}
 \eq
\beq\label{a11}
  \begin{array}{c}
  \displaystyle{
  \phi(z, u) f(z, -u)-\phi(z, -u) f(z, u)=\wp'(u)\,.
 }
 \end{array}
 \eq
We also need the identity
\beq\label{a12}
  \begin{array}{c}
  \displaystyle{
   \frac{1}{2}\,\frac{\wp'(z)-\wp'(w)}{\wp(z)-\wp(w)}=\zeta(z+w)-\zeta(z)-\zeta(w)=
   E_1(z+w)-E_1(z)-E_1(w)
 }
 \end{array}
 \eq
and
\beq\label{a13}
  \begin{array}{c}
  \displaystyle{
   \Big(\zeta(z+w)-\zeta(z)-\zeta(w)\Big)^2=\wp(z)+\wp(w)+\wp(z+w)\,.
 }
 \end{array}
 \eq
%

%%%%%%%%%%%%%%%%%%%%%%%%%%%%%%%%%%%%%%%%%%%%%%%%%%%%%%%%%%%%%%%%%%%%%%%%%%%%%%%%%%%%%%%%%%%
%%%%%%%%%%%%%%%%%%%%%%%%%%%%%%%%%%%%%%%%%%%%%%%%%%%%%%%%%%%%%%%%%%%%%%%%%%%%%%%%%%%%%%%%%%%
\section{Appendix B: elliptic $R$-matrices}\label{secB}
\def\theequation{B.\arabic{equation}}
\setcounter{equation}{0}

For the readers convenience
here we collect a brief description of $R$-matrices under consideration and their interrelations.

\subsection*{Quantum $R$-matrices}

\paragraph{Non-dynamical Baxter-Belavin $R$-matrix.} We begin with the elliptic Baxter-Belavin $R$-matrix \cite{BB}.
It is defined in the matrix basis (\ref{e86})-(\ref{e87}):
 \beq\label{e901}
 \begin{array}{c}
  \displaystyle{
 R^{\hbar}_{12}(z,w)=\frac{1}{N}\sum\limits_{\substack{a\in\,
 \mZ_{ N}\times\mZ_{ N} }} T_a\otimes T_{-a}
 \exp \Big(2\pi\imath\,\frac{a_2(z-w)}{N}\Big)\,\phi\Big(z-w,\frac{a_1+a_2\tau}{N}+\frac{\hbar}{N}\Big)
 }
 \end{array}
 \eq
and satisfies the quantum Yang-Baxter equation:
  \beq\label{e902}
  \begin{array}{c}
  \displaystyle{
 R^\hbar_{12}(z_1,z_2)R^\hbar_{13}(z_1,z_3)R^\hbar_{23}(z_2,z_3)
 =R^\hbar_{23}(z_2,z_3)R^\hbar_{13}(z_1,z_3)R^\hbar_{12}(z_1,z_2)\,.
 }
 \end{array}
 \eq
Besides \cite{BB} more properties for this $R$-matrix
can be found in \cite{Baxter2}, \cite{LOZ15} and in the Appendix of \cite{ZZ}.

\paragraph{Dynamical Felder's elliptic $R$-matrix.} The next is the quantum dynamical $R$-matrix
introduced by G. Felder \cite{Felder2}:
 \beq\label{e903}
 \begin{array}{c}
  \displaystyle{
 R^{\hbox{\tiny{F}}}_{12}(\hbar,z_1,z_2|\,q)= \sum\limits_{i,j:\,i\neq j}^N
 E_{ii}\otimes E_{jj}\, \phi(\hbar,-q_{ij})+
 }
\\
  \displaystyle{
+\sum\limits_{i,j:\,i\neq j}^N
 E_{ij}\otimes E_{ji}\, \phi(z_1-z_2,q_{ij})+\phi(\hbar,z_1-z_2)\sum\limits_{i=1}^N
 E_{ii}\otimes E_{ii}\,.
 }
 \end{array}
 \eq
It satisfies the quantum dynamical Yang-Baxter equation (or the
Gervais-Neveu-Felder equation):
  \beq\label{e904}
  \begin{array}{c}
  \displaystyle{
 R^{\hbox{\tiny{F}}}_{12}(\hbar,z_1,z_2|\,q)
 R^{\hbox{\tiny{F}}}_{13}(\hbar,z_1,z_3|\,q-\hbar^{(2)})R^{\hbox{\tiny{F}}}_{23}(\hbar,z_2,z_3|\,q)=\hspace{40mm}
}
\\ \ \\
  \displaystyle{
\hspace{40mm}
=R^{\hbox{\tiny{F}}}_{23}(\hbar,z_2,z_3|\,q-\hbar^{(1)})
R^{\hbox{\tiny{F}}}_{13}(\hbar,z_1,z_3|\,q)R^{\hbox{\tiny{F}}}_{12}(\hbar,z_1,z_2|\,q-\hbar^{(3)})\,.
 }
 \end{array}
 \eq
 The shifts of arguments $q_i$ are performed in the following way:
  \beq\label{e905}
  \begin{array}{c}
  \displaystyle{
R^{\hbox{\tiny{F}}}_{12}(\hbar,z_1,z_2|\,q+\hbar^{(3)})=P_3^\hbar\,
R^{\hbox{\tiny{F}}}_{12}(\hbar,z_1,z_2|\,q)\, P_3^{-\hbar} \,,\quad
P_3^\hbar=\sum\limits_{k=1}^N 1_N\otimes 1_N\otimes E_{kk}
\exp\Big(\hbar\frac{\p}{\p q_k}\Big)\,.
 }
 \end{array}
 \eq

 \paragraph{IRF-Vertex correspondence.}
 The relation between dynamical and non-dynamical $R$-matrices was found by R.J. Baxter \cite{Baxter2}
  and is called the IRF-Vertex correspondence (or the IRF-Vertex relation between quantum $R$-matrices):
 \beq\label{e906}
 \begin{array}{c}
  \displaystyle{
g_2(z_2,q)\,
g_1(z_1,q+\hbar^{(2)})\,R^{\hbox{\tiny{F}}}_{12}(\hbar,z_1-z_2|\,q)=R^\hbar_{12}(\hbar,z_1-z_2)
g_1(z_1,q)\, g_2(z_2,q+\hbar^{(1)})\,,
 }
 \end{array}
 \eq
 where $g(z,q)$ is the matrix (\ref{e89}) and the shifts of arguments are defined similarly to (\ref{e905}):
  \beq\label{e907}
  \begin{array}{c}
  \displaystyle{
g_1(z_1,q+\hbar^{(2)})=P_2^{\hbar}\,
g_1(z_1,q) P_2^{-\hbar}\in\Mat^{\otimes 2} \,,\quad
P_2^\hbar=\sum\limits_{k=1}^N 1_{N}\otimes E_{kk}
\exp\Big(\hbar\frac{\p}{\p q_k}\Big)\,.
 }
 \end{array}
 \eq

\subsection*{Classical $r$-matrices}
\paragraph{Non-dynamical Belavin-Drinfeld $r$-matrix.}
In the quasi-classical limit $\hbar\to 0$ the Baxter-Belavin
$R$-matrix (\ref{e901}) has the expansion
  \beq\label{e908}
  \begin{array}{c}
  \displaystyle{
R^\hbar(z_1,z_2)=\frac{1}{\hbar}\,1_N\otimes 1_N+r^{\hbox{\tiny{BD}}}_{12}(z_1,z_2)+O(\hbar)\,,
 }
 \end{array}
 \eq
where $r^{\hbox{\tiny{BD}}}_{12}(z_1,z_2)$ is the Belavin-Drinfeld $r$-matrix \cite{BD} (\ref{e85}).
The Belavin-Drinfeld $r$-matrix satisfies the classical Yang-Baxter equation
  \beq\label{e909}
  \begin{array}{c}
  \displaystyle{
[r^{\hbox{\tiny{BD}}}_{12}(z_1,z_2),r^{\hbox{\tiny{BD}}}_{13}(z_1,z_3)]+
[r^{\hbox{\tiny{BD}}}_{12}(z_1,z_2),r^{\hbox{\tiny{BD}}}_{23}(z_2,z_3)]+
[r^{\hbox{\tiny{BD}}}_{13}(z_1,z_3),r^{\hbox{\tiny{BD}}}_{23}(z_2,z_3)]=0\,.
 }
 \end{array}
 \eq
It is deduced from the quantum Yang-Baxter equation (\ref{e902}) in the limit (\ref{e908}).
Let us mention that from the classical mechanics viewpoint the Yang-Baxter equation (\ref{e909})
is a sufficient condition for the Jacobi identity $\{\{L_1(z_1),L_2(z_2)\},L_3(z_3)\}+{\rm cycl.}=0$
for non-dynamical skew-symmetric (i.e. $r_{12}(z_1,z_2)=-r_{21}(z_2,z_1)$) $r$-matrix and the classical
$r$-matrix structure
  \beq\label{e9091}
  \begin{array}{c}
  \displaystyle{
\{L_1(z_1),L_2(z_2)\}=[L_1(z_1)+L_2(z_2),r_{12}(z_1,z_2)]\,.
 }
 \end{array}
 \eq
In the general case (when non-dynamical $r$-matrix is not skew-symmetric) the classical
exchange relation is
  \beq\label{e9092}
  \begin{array}{c}
  \displaystyle{
\{L_1(z_1),L_2(z_2)\}=[L_1(z_1),r_{12}(z_1,z_2)]-[L_2(z_2),r_{21}(z_2,z_1)]
 }
 \end{array}
 \eq
and the classical Yang-Baxter equation has the form:
  \beq\label{e9093}
  \begin{array}{c}
  \displaystyle{
[r_{12}(z_1,z_2),r_{13}(z_1,z_3)]+
[r_{12}(z_1,z_2),r_{23}(z_2,z_3)]+
[r_{32}(z_3,z_2),r_{13}(z_1,z_3)]=0\,.
 }
 \end{array}
 \eq
For skew-symmetric non-dynamical $r$-matrices (\ref{e9093}) turns into (\ref{e909})
and (\ref{e9092}) turns into (\ref{e9091}).

\paragraph{Dynamical $r$-matrices: classical Felder's $r$-matrix.} Consider first the quasi-classical limit of
the Felder's $R$-matrix (\ref{e903}):
  \beq\label{e910}
  \begin{array}{c}
  \displaystyle{
R^{\hbox{\tiny{F}}}(\hbar,z_1,z_2|\, q)=\frac{1}{\hbar}\,1_N\otimes 1_N
+r^{\hbox{\tiny{F}}}_{12}(z_1,z_2)+O(\hbar)\,,
 }
 \end{array}
 \eq
where
\beq\label{e911}
  \begin{array}{c}
  \displaystyle{
 r^{\hbox{\tiny{F}}}_{12}(z,w)=
 E_1(z-w)\,\sum\limits_{i=1}^NE_{ii}\otimes E_{ii}+
 \sum\limits^N_{i\neq j}\phi(z-w,q_{ij})\,E_{ij}\otimes E_{ji}
 -\sum\limits^N_{i\neq j}E_1(q_{ij})E_{ii}\otimes E_{jj}\,.
 }
 \end{array}
 \eq
In the limit $\hbar \to 0$ the quantum relation (\ref{e904}) provides
the classical dynamical Yang-Baxter equation:
  \beq\label{e912}
  \begin{array}{c}
  \displaystyle{
[r^{\hbox{\tiny{F}}}_{12}(z_1,z_2),r^{\hbox{\tiny{F}}}_{13}(z_1,z_3)]+
[r^{\hbox{\tiny{F}}}_{12}(z_1,z_2),r^{\hbox{\tiny{F}}}_{23}(z_2,z_3)]+
[r^{\hbox{\tiny{F}}}_{13}(z_1,z_3),r^{\hbox{\tiny{F}}}_{23}(z_2,z_3)]+
}
\\ \ \\
  \displaystyle{
  +[\mD^{(1)},r^{\hbox{\tiny{F}}}_{23}(z_2,z_3)]
  -[\mD^{(2)},r^{\hbox{\tiny{F}}}_{13}(z_1,z_3)]
  +[\mD^{(3)},r^{\hbox{\tiny{F}}}_{12}(z_1,z_2)]=0
\,,
 }
 \end{array}
 \eq
where
\beq\label{e913}
  \begin{array}{c}
  \displaystyle{
 \mD^{(1)}=\sum\limits_{k=1}^N E_{kk}\otimes 1_N\otimes 1_N
\frac{\p}{\p q_k}\,,
\quad
 \mD^{(2)}=\sum\limits_{k=1}^N 1_N\otimes  E_{kk}\otimes 1_N
\frac{\p}{\p q_k}
 }
 \end{array}
 \eq
and similarly for $\mD^{(3)}$. Notice also that
\beq\label{e914}
  \begin{array}{c}
  \displaystyle{
 [\mD^{(1)},r^{\hbox{\tiny{F}}}_{23}(z_2,z_3)]
 %=\{L_1(z_1),r^{\hbox{\tiny{F}}}_{23}(z_2,z_3)\}
 =\{P_1,r^{\hbox{\tiny{F}}}_{23}(z_2,z_3)\}\,,
 }
 \end{array}
 \eq
where $P_1=P\otimes 1_N\otimes 1_N$, $P={\rm diag}(p_1,...,p_N)\in\Mat$. So that
(\ref{e912}) takes the form:
  \beq\label{e915}
  \begin{array}{c}
  \displaystyle{
[r^{\hbox{\tiny{F}}}_{12}(z_1,z_2),r^{\hbox{\tiny{F}}}_{13}(z_1,z_3)]+
[r^{\hbox{\tiny{F}}}_{12}(z_1,z_2),r^{\hbox{\tiny{F}}}_{23}(z_2,z_3)]+
[r^{\hbox{\tiny{F}}}_{13}(z_1,z_3),r^{\hbox{\tiny{F}}}_{23}(z_2,z_3)]+
}
\\ \ \\
  \displaystyle{
  +\{P_1,r^{\hbox{\tiny{F}}}_{23}(z_2,z_3)\}
  -\{P_2,r^{\hbox{\tiny{F}}}_{13}(z_1,z_3)\}
  +\{P_3,r^{\hbox{\tiny{F}}}_{12}(z_1,z_2)\}=0
\,.
 }
 \end{array}
 \eq
\paragraph{Dynamical $r$-matrix for the (multi)spin Calogero-Moser model.}
In this paper we use the classical dynamical $r$-matrix $r^{\hbox{\tiny{spin}}}_{12}(z,w)$
(\ref{e40}) for description of the spin and multispin Calogero-Moser models.
It is related to the classical Felder's $r$-matrix (\ref{e911}) as follows:
\beq\label{e916}
  \begin{array}{c}
  \displaystyle{
 r^{\hbox{\tiny{spin}}}_{12}(z,w)=r^{\hbox{\tiny{F}}}_{12}(z,w)+\delta r_{12}(z,w)\,,
 \qquad
 \delta r_{12}(z,w)=\sum\limits^N_{i\neq j}E_1(q_{ij})E_{ii}\otimes E_{jj}\,.
 }
 \end{array}
 \eq
The difference $\delta r_{12}(z,w)$ between two $r$-matrices is in fact a twist.
See e.g. Lemma 1 in \cite{LOSZ}. The latter means that both $r$-matrices $r^{\hbox{\tiny{F}}}_{12}(z,w)$
and $r^{\hbox{\tiny{spin}}}_{12}(z,w)$
satisfy the classical dynamical Yang-Baxter equation (\ref{e915}).

\paragraph{Dynamical $r$-matrix for the spinless Calogero-Moser model.}
For description of the spinless Calogero-Moser model we use the $r$-matrix (\ref{e26}) \cite{Skl3,BradenSuz}.
The Jacobi identity for the Lax matrix
\beq\label{e917}
  \begin{array}{c}
  \displaystyle{
 \{\{L_1(z_1),L_2(z_2)\},L_3(z_3)\}+{\rm cycl.}=0
 }
 \end{array}
 \eq
 is rewritten in the form
\beq\label{e918}
  \begin{array}{c}
  \displaystyle{
 [\mR_{123},L_1(z_1)]+[\mR_{231},L_2(z_2)]+[\mR_{312},L_3(z_3)]=0\,,
 }
 \end{array}
 \eq
 where
\beq\label{e919}
  \begin{array}{c}
  \displaystyle{
\mR_{123}=[r^{\hbox{\tiny{CM}}}_{12}(z_1,z_2),r^{\hbox{\tiny{CM}}}_{13}(z_1,z_3)]+
[r^{\hbox{\tiny{CM}}}_{12}(z_1,z_2),r^{\hbox{\tiny{CM}}}_{23}(z_2,z_3)]+
[r^{\hbox{\tiny{CM}}}_{32}(z_3,z_2),r^{\hbox{\tiny{CM}}}_{13}(z_1,z_3)]-
}
\\ \ \\
  \displaystyle{
  -\{P_2,r^{\hbox{\tiny{CM}}}_{13}(z_1,z_3)\}
  +\{P_3,r^{\hbox{\tiny{CM}}}_{12}(z_1,z_2)\}\in\Mat^{\otimes 3}\,,
 }
 \end{array}
 \eq
 and $\mR_{123}$ does not vanish. It is equal to a certain expression. See details in \cite{Skl3}.
 In principle, this expression can be deduced by representing $r^{\hbox{\tiny{CM}}}_{12}(z_1,z_2)$ in the form
 $r^{\hbox{\tiny{CM}}}_{12}(z_1,z_2)=r^{\hbox{\tiny{spin}}}_{12}(z_1,z_2)+\Delta r_{12}(z_2)$ and
 using (\ref{e915}) for $r^{\hbox{\tiny{spin}}}_{12}(z_1,z_2)$.

We also use the modified version of the $r$-matrix $r^{\hbox{\tiny{CM}}}_{12}(z_1,z_2)$
related to the Lie algebra ${\rm sl}_N$:
\beq\label{e920}
  \begin{array}{c}
  \displaystyle{
 {\tilde r}^{\hbox{\tiny{CM}}}_{12}(z,w)=r^{\hbox{\tiny{CM}}}_{12}(z,w)-l_2(w)\,,
 }
 \end{array}
 \eq
where $l_2(w)=1_N\otimes l(w)$ with
\beq\label{e921}
  \begin{array}{c}
  \displaystyle{
 l(w)=g^{-1}(w,q)\p_w g(w,q)\in\Mat
 }
 \end{array}
 \eq
 or, explicitly\footnote{Explicit expression (\ref{e922}) is related to factorization of Lax matrices.
 See \cite{VZ} and \cite{ZZ}.}
\beq\label{e922}
  \begin{array}{c}
  \displaystyle{
 l_{ij}(w)=\frac{\delta_{ij}}{N}\Big(E_1(w)-\sum\limits_{k\neq i}^N E_1(q_{ik})\Big)
 +\frac{1-\delta_{ij}}{N}\,\phi(w,q_{ij})\,.
 }
 \end{array}
 \eq
 Notice that by definition (\ref{e920}) the $r$-matrix ${\tilde r}^{\hbox{\tiny{CM}}}_{12}(z,w)$
 also satisfies the classical exchange relation (\ref{e04}) since $[L_1(z),l_2(w)]=0$.

 Finally, for the IRF-Vertex relation in the field theory case we use
\beq\label{e9201}
  \begin{array}{c}
  \displaystyle{
 {\bar r}^{\hbox{\tiny{CM}}}_{12}(z,w)=r^{\hbox{\tiny{CM}}}_{12}(z,w)-{\bar l}_2(w)\,,
 }
 \end{array}
 \eq
 where ${\bar l}_2(w)=1_N\otimes {\bar l}(w)$ with ${\bar l}(w)\in\Mat$ as follows:
\beq\label{e9221}
  \begin{array}{c}
  \displaystyle{
 {\bar l}_{ij}(w)=\frac{\delta_{ij}}{N}\,E_1(w)
 +\frac{1-\delta_{ij}}{N}\,\phi(w,q_{ij})\,.
 }
 \end{array}
 \eq

\subsection*{Classical IRF-Vertex relations}
In the quasi-classical limits $\hbar\to 0$ (\ref{e908}) and (\ref{e910}) the
IRF-Vertex correspondence (\ref{e906}) provides the following relation between
the classical dynamical $r$-matrix (\ref{e911}) and the Belavin-Drinfeld non-dynamical
$r$-matrix (\ref{e85}):
\beq\label{e923}
  \begin{array}{c}
  \displaystyle{
 r^{\hbox{\tiny{BD}}}_{12}(z,w)=
 g_1(z,q)g_2(w,q)r^{\hbox{\tiny{F}}}_{12}(z,w)g_1^{-1}(z,q)g_2^{-1}(w,q)-
 }
 \\ \ \\
   \displaystyle{
 -g_2(w,q)\{g_1(z,q),P_2\}g_1^{-1}(z,q)g_2^{-1}(w,q)
 -g_1(z,q)\{P_1,g_2(w,q)\}g_1^{-1}(z,q)g_2^{-1}(w,q)
 }
 \end{array}
 \eq
 with the matrix $g(z,q)$ (\ref{e89})-(\ref{e91}).

The gauge equivalence between Euler-Arnold tops and Calogero-Moser models \cite{LOZ,VZ,AtZ1,AtZ3}
provides the following relation between the classical $r$-matrices (\ref{e85}) and (\ref{e920}):
\beq\label{e924}
  \begin{array}{c}
  \displaystyle{
 r^{\hbox{\tiny{BD}}}_{12}(z,w)=
 g_1(z,q)g_2(w,q){\tilde r}^{\hbox{\tiny{CM}}}_{12}(z,w)g_1^{-1}(z,q)g_2^{-1}(w,q)-
 }
 \\ \ \\
   \displaystyle{
 -g_2(w,q)\{g_1(z,q),P_2\}g_1^{-1}(z,q)g_2^{-1}(w,q)
 }
 \end{array}
 \eq
or
\beq\label{e925}
  \begin{array}{c}
  \displaystyle{
 r^{\hbox{\tiny{BD}}}_{12}(z,w)=
 g_1(z,q)g_2(w,q)\Big({r}^{\hbox{\tiny{CM}}}_{12}(z,w)-l_2(w)\Big)g_1^{-1}(z,q)g_2^{-1}(w,q)-
 }
 \\ \ \\
   \displaystyle{
 -g_2(w,q)\{g_1(z,q),P_2\}g_1^{-1}(z,q)g_2^{-1}(w,q)\,.
 }
 \end{array}
 \eq

Finally, consider the normalized gauge transformation
\beq\label{e926}
  \begin{array}{c}
  \displaystyle{
 {\bar g}(z,q)=\frac{1}{(\det g(z,q))^{1/N}}\,g(z,q)\,,\qquad \det{\bar g}(z,q)=1\,.
 }
 \end{array}
 \eq
Due to (see e.g. \cite{LOZ,ZZ})
\beq\label{e927}
  \begin{array}{c}
  \displaystyle{
 \det\Xi(z,q)=c_N(\tau)\vth(z)V(q)\,,\quad V(q)=\prod\limits_{i<j}\vth(q_{ij})
 }
 \end{array}
 \eq
with some constant $c_N(\tau)$ and $\det d^0=(-1)^{N(N-1)/2}V^2(q)$,
we have
\beq\label{e928}
  \begin{array}{c}
  \displaystyle{
 \det g(z,q)=c_N(\tau)\frac{\vth(z)}{V(q)}\,.
 }
 \end{array}
 \eq
 Then
\beq\label{e929}
  \begin{array}{c}
  \displaystyle{
 r^{\hbox{\tiny{BD}}}_{12}(z,w)=
 {\bar g}_1(z,q){\bar g}_2(w,q)\Big({r}^{\hbox{\tiny{CM}}}_{12}(z,w)-{\bar l}_2(w)\Big){\bar g}_1^{-1}(z,q){\bar g}_2^{-1}(w,q)-
 }
 \\ \ \\
   \displaystyle{
 -{\bar g}_2(w,q)\{{\bar g}_1(z,q),P_2\}{\bar g}_1^{-1}(z,q){\bar g}_2^{-1}(w,q)
 }
 \end{array}
 \eq
 or
\beq\label{e930}
  \begin{array}{c}
  \displaystyle{
 r^{\hbox{\tiny{BD}}}_{12}(z,w)=
 {\bar g}_1(z,q){\bar g}_2(w,q){\bar r}^{\hbox{\tiny{CM}}}_{12}(z,w){\bar g}_1^{-1}(z,q){\bar g}_2^{-1}(w,q)-
 }
 \\ \ \\
   \displaystyle{
 -{\bar g}_2(w,q)\{{\bar g}_1(z,q),P_2\}{\bar g}_1^{-1}(z,q){\bar g}_2^{-1}(w,q)\,,
 }
 \end{array}
 \eq
 so that the difference $l_2(w)$ in (\ref{e925}) and ${\bar l}_2(2)$ in (\ref{e930})
 is compensated by the difference between $\{{g}_1(z,q),P_2\}$ and $\{{\bar g}_1(z,q),P_2\}$.

%%%%%%%%%%%%%%%%%%%%%%%%%%%%%%%%%%%%%%%%%%%%%%%%%%%%%%%%%%%%%%%%%%%%%%%%%%%%%%%%%%%%%%%%%%%%%%%%%%%%%%%%%%
%%%%%%%%%%%%%%%%%%%%%%%%%%%%%%%%%%%%%%%%%%%%%%%%%%%%%%%%%%%%%%%%%%%%%%%%%%%%%%%%%%%%%%%%%%%%%%
%%%%%%%%%%%%%%%%%%%%%%%%%%%%%%%%%%%%%%%%%%%%%%%%%%%%%%%%%%%%%%%%%%%%%%%%%%%%%%%%%%%%%%%%%%%%%%

\subsection*{Acknowledgments}

%\addcontentsline{toc}{section}{\hspace{6mm}Acknowledgments}

%This work was supported by the Russian Science Foundation under grant no. 19-11-00062,\\ %https://rscf.ru/en/project/19-11-00062/ .

We are grateful to K. Atalikov, A. Levin, M. Olshanetsky and A. Zabrodin for discussions.

This work was performed at the Steklov International Mathematical Center and supported by the Ministry of Science and Higher Education of the Russian Federation (agreement no. 075-15-2022-265).

%%%%%%%%%%%%%%%%%%%%%%%%%%%%%%%%%%%%%%%%%%%%%%%%%%%%%%%%%%%%%%%%%%%%%%%%%%%%%%%%%%%%%%%%%%%%%%
%%%%%%%%%%%%%%%%%%%%%%%%%%%%%%%%%%%%%%%%%%%%%%%%%%%%%%%%%%%%%%%%%%%%%%%%%%%%%%%%%%%%%%%%%%%%%%

\begin{small}

\end{small}

\end{document}